\documentclass[a4paper,nofootinbib,superscriptaddress]{revtex4}

\usepackage{color}
\usepackage{amsmath,amsfonts,amssymb,amsthm}
\usepackage{graphicx}
\usepackage{fullpage}
\PassOptionsToPackage{caption=false}{subfig}
\usepackage{subfig}

\usepackage{tikz}
\usetikzlibrary{arrows,decorations.pathmorphing,backgrounds,positioning,fit}
\usepackage{epstopdf} 
\usepackage{bm}  

\usepackage[pdfpagelabels,pdftex,bookmarks,breaklinks]{hyperref}
\definecolor{darkblue}{RGB}{0,0,127} 
\definecolor{darkgreen}{RGB}{0,150,0}
\hypersetup{colorlinks, linkcolor=darkblue, citecolor=darkgreen, filecolor=red, urlcolor=blue}

\usepackage{layout}

\pgfrealjobname{main}

\newtheorem{theorem}{Theorem}

\renewenvironment{proof}[1][Proof]{\noindent\textbf{#1.} }{\ $\Box$}

\def\Z{\mathbb{Z}}

\def\C{\mathbb{C}}

\newcommand{\Eref}[1]{Eq.~(\ref{#1})}
\newcommand{\Sref}[1]{Sec.~\ref{#1}}
\newcommand{\Fref}[1]{Fig.~\ref{#1}}
\newcommand{\Aref}[1]{Appendix~\ref{#1}}

\newcommand{\ket}[1]{|{#1}\rangle}

\newcommand{\ketbra}[2]{|{#1}\rangle\!\langle{#2}|}
\newcommand{\braket}[2]{\langle{#1}|{#2}\rangle}
\newcommand{\proj}[1]{\ketbra{#1}{#1}}

\newcommand{\subfont}[1]{\mathsf{#1}}

\newcommand{\bond}[3][b]{
	\begin{tikzpicture}[scale=.7,baseline=-1mm,auto]
		\node (Lprime) at  (0.45,0) [label=left:$#2$]{}; 
		\node (Rprime) at (1.2,0) [label=right:$#3$]{}; 
		\node (L) at  (0,0){}; 
		\node (R) at (1.6,0){}; 
		\draw[decorate, decoration=
			{snake,amplitude=.6mm,segment length=1.6mm}]
			(L) to node {{\scriptsize $b$}} (R);
	\end{tikzpicture}
}

\newcommand{\siteT}[4]{
	\begin{tikzpicture}[thick,scale=.7,baseline=-1pt]
		\draw (0,0) circle (30pt); 
		\path node at (0,0) [anchor=south east] {$#1$}
			node at (0,0) [anchor=south west] {$#2$}
			node at (0,0) [anchor=north west] {$#3$}
			node at (0,0) [anchor=north east] {$#4$};
		\draw (0,-1.4) -- (0,1.4);
		\node at (0,-30pt) [anchor=north west] {$e$};
	\end{tikzpicture}
}

\newcommand{\siteTnoe}[4]{
	\begin{tikzpicture}[thick,scale=.7,baseline=-1pt]
		\draw (0,0) circle (30pt); 
		\path node at (0,0) [anchor=south east] {$#1$}
			node at (0,0) [anchor=south west] {$#2$}
			node at (0,0) [anchor=north west] {$#3$}
			node at (0,0) [anchor=north east] {$#4$};
		\draw (0,-1.4) -- (0,1.4);
	\end{tikzpicture}
}

\newcommand{\siteTnoeB}[4]{
	\begin{tikzpicture}[thick,scale=.7,baseline=-1pt]
		\draw (0,0) circle (32pt); 
		\path node at (0,0) [anchor=south east] {$#1$}
			node at (0,0) [anchor=south west] {$#2$}
			node at (0,0) [anchor=north west] {$#3$}
			node at (0,0) [anchor=north east] {$#4$};
		\draw (0,-1.42) -- (0,1.42);
	\end{tikzpicture}
}

\newcommand{\siteA}[4]{
	\;
	\begin{tikzpicture}[thick,scale=.7,baseline=-1pt]
		\draw (0,0) circle (33pt); 
		\path node at (0,0) [anchor=south east] {$#1$}
			node at (0,0) [anchor=south west] {$#2$}
			node at (0,0) [anchor=north west] {$#3$}
			node at (0,0) [anchor=north east] {$#4$};
		\draw[-stealth'] (0,-1.5) -- (0,0.1);
		\draw (0,0) -- (0,1.5);
		\node at (0,-33pt) [anchor=north west] {$e$};
	\end{tikzpicture}
}

\newcommand{\siteAb}[4]{
	\begin{tikzpicture}[thick,scale=.7,baseline=-1pt]
		\draw (0,0) circle (38pt); 
		\path node at (0,0) [anchor=south east] {$#1$}
			node at (0,0) [anchor=south west] {$#2$}
			node at (0,0) [anchor=north west] {$#3$}
			node at (0,0) [anchor=north east] {$#4$};
		\draw[-stealth'] (0,-1.5) -- (0,0.1);
		\draw (0,0) -- (0,1.5);
		\node at (0,-36pt) [anchor=north west] {$e$};
	\end{tikzpicture}
}

\newcommand{\siteN}[4]{
	\left(\begin{array}{c}#1\\#4\end{array}
		\begin{tikzpicture}[baseline=-0.55]
		\draw[-stealth'] (0,-0.55) -- (0,0.1);
		\draw (0,0) -- (0,0.7);
		\end{tikzpicture}
	\begin{array}{c}#2\\#3\end{array}\right)
}
\newcommand{\siteNXL}[4]{
\siteN{#1}{#2}{#3}{#4}
}
\newcommand{\siteGIGANTOR}[4]{
\siteN{#1}{#2}{#3}{#4}
}

\newcommand{\vertexNA}[9]{\;
	\begin{tikzpicture}
		[scale=4, baseline=19mm]
		\node at (.5,.5) [anchor = south east] {$v$};

		\def\off{0.45}
		\foreach \a in {0,90,180,270}{
			\draw[->,>=stealth',draw=red,line width=2.5,rotate around={\a:(.5,.5)}](0.5,\off+.55+#9*0.05)-- (0.5,0.55+\off-#9*0.05);
		}

		\def\g{.4}
		\draw[red,ultra thick,shift={(.5,.5)}] (-1+\g,-1+\g) grid (1-\g,1-\g);
	
		\def\tick{0.14}
		\def\compsize{1.25mm}
		\def\identsize{0.75mm}
		\filldraw[draw=red, ultra thick, fill=white, rotate around={-90:(.5,.5)}]
			(\tick,0.5-\tick) circle (\compsize) node at (\tick,0.5-\tick) {$#1$};
		\filldraw[draw=red, ultra thick, fill=white, rotate around={-90:(.5,.5)}]
			(\tick,0.5+\tick) circle (\compsize) node at (\tick,0.5+\tick) {$#2$};
		\filldraw[draw=red, ultra thick, fill=white, rotate around={180:(.5,.5)}]
			(\tick,0.5-\tick) circle (\compsize) node at (\tick,0.5-\tick) {$#3$};
		\filldraw[draw=red, ultra thick, fill=white, rotate around={180:(.5,.5)}]
			(\tick,0.5+\tick) circle (\compsize) node at (\tick,0.5+\tick) {$#4$};
		\filldraw[draw=red, ultra thick, fill=white, rotate around={-270:(.5,.5)}]
			(\tick,0.5-\tick) circle (\compsize) node at (\tick,0.5-\tick) {$#5$};
		\filldraw[draw=red, ultra thick, fill=white, rotate around={-270:(.5,.5)}]
			(\tick,0.5+\tick) circle (\compsize) node at (\tick,0.5+\tick) {$#6$};
		\filldraw[draw=red, ultra thick, fill=white, rotate around={0:(.5,.5)}]
			(\tick,0.5-\tick) circle (\compsize) node at (\tick,0.5-\tick) {$#7$};
		\filldraw[draw=red, ultra thick, fill=white, rotate around={0:(.5,.5)}]
			(\tick,0.5+\tick) circle (\compsize) node at (\tick,0.5+\tick) {$#8$};
	\end{tikzpicture}\;
}

\newcommand{\plaquetteNA}[9]{\;
	\begin{tikzpicture}
		[scale=4, baseline=18mm]
		\node at (0.5,0.5) [anchor = center] {$p$};

		\def\bigc{4.5mm}
		\def\off{0.45}
		\def\arroff{0.12}
		\foreach \a in{0,180}
			\draw[->,>=stealth',draw=blue,line width=2.5,rotate around={\a:(.5,.5)}] (-\arroff,0.05+\off -#9*0.05)--(-\arroff,\off+.05+#9*0.05);
		\foreach \a in{90,270}	
			\draw[->,>=stealth',draw=blue,line width=2.5,rotate around={\a:(.5,.5)}] (-\arroff,.05+#9*0.05+\off)--(-\arroff,\off+.05-#9*0.05);

		\def\g{.1}
		\def\offs{1.69mm}
		\draw[blue,ultra thick,shift={(.5,.5)}] (-\offs-\bigc,-\offs-\bigc) rectangle (\offs+\bigc,\offs+\bigc);
	
		\def\identsize{0.75mm}
		\def\tick{0.06}
		\def\ticka{0.18}
		\def\compsize{1.7mm}
		\def\identoff{-0.145}
		\filldraw[draw=blue, ultra thick, fill=white, rotate around={-90:(.5,.5)}]
			(\tick,0.5-\ticka) circle (\compsize) node at (\tick,0.5-\ticka) {$#1$};
		\filldraw[draw=blue, ultra thick, fill=white, rotate around={-90:(.5,.5)}]
			(\tick,0.5+\ticka) circle (\compsize) node at (\tick,0.5+\ticka) {$#2$};
		\filldraw[draw=blue, ultra thick, fill=white, rotate around={180:(.5,.5)}]
			(\tick,0.5-\ticka) circle (\compsize) node at (\tick,0.5-\ticka) {$#3$};
		\filldraw[draw=blue, ultra thick, fill=white, rotate around={180:(.5,.5)}]
			(\tick,0.5+\ticka) circle (\compsize) node at (\tick,0.5+\ticka) {$#4$};
		\filldraw[draw=blue, ultra thick, fill=white, rotate around={-270:(.5,.5)}]
			(\tick,0.5-\ticka) circle (\compsize) node at (\tick,0.5-\ticka) {$#5$};
		\filldraw[draw=blue, ultra thick, fill=white, rotate around={-270:(.5,.5)}]
			(\tick,0.5+\ticka) circle (\compsize) node at (\tick,0.5+\ticka) {$#6$};
		\filldraw[draw=blue, ultra thick, fill=white, rotate around={0:(.5,.5)}]
			(\tick,0.5-\ticka) circle (\compsize) node at (\tick,0.5-\ticka) {$#7$};
		\filldraw[draw=blue, ultra thick, fill=white, rotate around={0:(.5,.5)}]
			(\tick,0.5+\ticka) circle (\compsize) node at (\tick,0.5+\ticka) {$#8$};
	\end{tikzpicture}\;
}

\newcommand{\vertexNAenc}[5]{
\begin{tikzpicture}[scale=2.2, baseline=10mm]
	\node at (.48,.52) [anchor = south east] {$v$};

	\def\g{.3}
	\draw[red,ultra thick,shift={(.5,.5)}] (-1+\g,-1+\g) grid (1-\g,1-\g);

	\def\off{0.15}
	\foreach \a in {0,90,180,270}{
		\draw[->,>=stealth',draw=red,line width=2,rotate around={\a:(.5,.5)}] (0.5+#5*0.1+\off,0.5)--(0.5+\off-#5*0.1,0.5);
	}

	\filldraw[draw=red, ultra thick, fill=white, rotate around={-90:(.5,.5)}]
		(0,.5) circle (2mm) node at (0,.5) {$#1$};
	\filldraw[draw=red, ultra thick, fill=white, rotate around={0:(.5,.5)}]
		(0,.5) circle (2mm) node at (0,.5) {$#4$};
	\filldraw[draw=red, ultra thick, fill=white, rotate around={90:(.5,.5)}]
		(0,.5) circle (2mm) node at (0,.5) {$#3$};
	\filldraw[draw=red, ultra thick, fill=white, rotate around={180:(.5,.5)}]
		(0,.5) circle (2mm) node at (0,.5) {$#2$};
\end{tikzpicture}
}

\newcommand{\plaquetteNAenc}[5]{
\begin{tikzpicture}[scale=2.2, baseline=10mm]
	\node at (.5,.5) {$p$};

	\def\g{.013}
	\draw[blue,ultra thick] (-\g,-\g) grid (1+\g,1+\g);

	\foreach \a in {0,90,180,270}
		\filldraw[draw=blue, ultra thick, fill=white, rotate around={\a:(.5,.5)}] 
			(0,.5) circle (2mm);

	\def\off{0.15}
	\foreach \a in {0,180}{
		\draw[->,>=stealth',draw=blue,line width=2,rotate around={\a:(.5,.5)}] (-#5*0.1+\off,0)--(\off+#5*0.1,0);
		\draw[->,>=stealth',draw=blue,line width=2,rotate around={\a:(.5,.5)}] (0,-#5*0.1+\off)--(0,\off+#5*0.1);
	}
			
	\path node at (0,.5) {$#4$} node at (1,.5) {$#2$} 
		node at (.5,0) {$#3$} node at (.5,1) {$#1$};
\end{tikzpicture}
}

\newcommand{\vertexNAencNoArrow}[4]{
\begin{tikzpicture}[scale=2.2, baseline=10mm]
	\node at (.48,.52) [anchor = south east] {$v$};

	\def\g{.3}
	\draw[red,ultra thick,shift={(.5,.5)}] (-1+\g,-1+\g) grid (1-\g,1-\g);

	\filldraw[draw=red, ultra thick, fill=white, rotate around={-90:(.5,.5)}]
		(0,.5) circle (2mm) node at (0,.5) {$#1$};
	\filldraw[draw=red, ultra thick, fill=white, rotate around={0:(.5,.5)}]
		(0,.5) circle (2mm) node at (0,.5) {$#4$};
	\filldraw[draw=red, ultra thick, fill=white, rotate around={90:(.5,.5)}]
		(0,.5) circle (2mm) node at (0,.5) {$#3$};
	\filldraw[draw=red, ultra thick, fill=white, rotate around={180:(.5,.5)}]
		(0,.5) circle (2mm) node at (0,.5) {$#2$};
\end{tikzpicture}
}

\newcommand{\plaquetteNAencNoArrow}[4]{
\begin{tikzpicture}[scale=2.2, baseline=10mm]
	\node at (.5,.5) {$p$};

	\def\g{.013}
	\draw[blue,ultra thick] (-\g,-\g) grid (1+\g,1+\g);

	\foreach \a in {0,90,180,270}
		\filldraw[draw=blue, ultra thick, fill=white, rotate around={\a:(.5,.5)}] 
			(0,.5) circle (2mm);
			
	\path node at (0,.5) {$#4$} node at (1,.5) {$#2$} 
		node at (.5,0) {$#3$} node at (.5,1) {$#1$};
\end{tikzpicture}
}

\begin{document}

\title{Toric codes and quantum doubles from two-body Hamiltonians}

\author{Courtney G.\ Brell}
\affiliation{Centre for Engineered Quantum Systems, School of Physics, The University of Sydney, Sydney, Australia}

\author{Steven T.\ Flammia}
\affiliation{Perimeter Institute for Theoretical Physics, 
	Waterloo, 
	Canada}
\affiliation{California Institute of Technology, 
	Institute for Quantum Information, 
	Pasadena, 
	USA}

\author{Stephen D.\ Bartlett}
\affiliation{Centre for Engineered Quantum Systems, School of Physics, The University of Sydney, Sydney, Australia}
\author{Andrew C.\ Doherty}
\affiliation{Centre for Engineered Quantum Systems, School of Physics, The University of Sydney, Sydney, Australia}
\affiliation{School of Mathematics and Physics, The University of Queensland, St Lucia, Australia}


\begin{abstract}
We present a procedure to obtain the Hamiltonians of the toric code and Kitaev quantum double models as the low-energy limits of entirely two-body Hamiltonians. Our construction makes use of a new type of perturbation gadget based on error-detecting subsystem codes. The procedure is motivated by a PEPS description of the target models, and reproduces the target models' behavior using only couplings which are natural in terms of the original Hamiltonians. This allows our construction to capture the symmetries of the target models.
\end{abstract}

\maketitle

\vspace{-25pt} 
\tableofcontents
\clearpage 



%
%

\section{Introduction}\label{S:intro}

There has been a surge of interest recently in spin lattice models because of their strange and wonderful properties.  From the perspective of condensed matter physics and quantum many-body theory, they have recently led to major advances in our understanding of the nature of quantum phase transitions and topological order in two-dimentional systems. The topological properties of these models are also of interest in quantum computing and quantum error correction. A system with topological order can possess intrinsic error correction or protection capabilities. These are exploited for quantum data storage~\cite{Dennis2002,Wang2003,Duclos-Cianci2010} or quantum information processing~\cite{Raussendorf98,Raussendorf2006,Raussendorf2007} with high error thresholds.
The encoded logical operations in topological models are associated with non-trivial homology cycles on a lattice of spins.  A lattice which has a non-trivial topology (such as a torus or punctured disk) can encode quantum information into its ground states which is robust to small local perturbations of the Hamiltonian.

Many models which are relatively simple (from a theoretical point of view) contain topologically ordered ground states. The toric code and its generalization to the quantum double models~\cite{KitaevTC97} are a significant class of exactly solvable models containing a range of different topological orders. Quantum double models of non-Abelian groups can support non-Abelian anyons, quasiparticles whose exchange statistics transcend the traditional Bose-Einstein/Fermi-Dirac dichotomy that is ubiquitous in three dimensions.  Braiding such quasiparticles can be used for universal quantum computation~\cite{Mochon2003,Mochon2004,Preskill1997,Ogburn1999}. 

However, these models (and other more general topological models~\cite{LevinSN05}) consist of many-body interactions that are quite challenging to implement experimentally, as they usually involve interactions between four or more bodies.  By contrast, most natural couplings are only 2-body.  It is therefore of interest to find systems with only 2-body interactions which can realize topologically ordered phases. 

An example of a 2-body system with a topologically-ordered ground state is the Kitaev honeycomb model~\cite{Kitaev-06}. This well-studied model is being pursued experimentally in a number of systems, but unfortunately it cannot be used for universal quantum computation.

Aside from explicitly finding 2-body models which reproduce a particular desired type of topological order (certainly challenging), one can use perturbative techniques to reproduce an existing many-body model as the low energy effective behaviour of a 2-body system. The perturbative gadgets approach~\cite{Kempe2006,Bravyi2008a,Jordan2008,Oliveira2008} is the standard tool to achieve this, but it has a number of drawbacks. By tailoring the perturbation gadgets to specific classes of models, one might hope the result is a simpler construction circumventing many of these difficulties. 

Here we present a new type of perturbation gadget that works by encoding the logical qudits of the target models in quantum error-detecting codes. This allows us to reproduce the properties of topological models as the low-energy effective Hamiltonians of 2-body systems. Here we concentrate specifically on the quantum double models, but we anticipate that a similar mechanism could be tailored to other classes of models (e.g. string net models~\cite{LevinSN05}). Our construction is natural, in the sense that all of the interactions of our system are very closely related to the interactions of the target model, and because of this, an extensive number of symmetries of the target model are preserved exactly from the level of the physical lattice. 

Unlike Kitaev's honeycomb model, our constructions are not exactly solvable. However, our results are a significant extension of Kitaev's method in that they can yield any type of topological order (i.e.~different types of anyons) within the class spanned by the quantum double models, including those which are universal for quantum computation.

%
%

\section{Results and Methods}

In this section we give an overview of our results and the methods we use to obtain them. To avoid obscuring the essence of our work with unnecessary technical details, we will use the toric code model as a concrete example in many places.  However, we stress that our results immediately carry over to all of the cyclic ($\Z_d$) quantum double models, and extend to the general non-Abelian quantum double models with just minor adjustments.  

\subsection{The Quantum Double Models}

The quantum double models~\cite{KitaevTC97} are a class of spin-lattice models which exhibit topological order. They can be used as topological quantum error-correcting codes based on the algebra of the Drinfeld double $\mathcal{D}(G)$ of a group $G$. The simplest member, corresponding to the group $\Z_2$, is the well-studied toric code model.

For simplicity, we define the quantum double models on a square lattice (although they can be defined on any oriented graph) with qudits ($d$-level quantum systems) on the edges, as in Figure~\ref{F:ToricOps}.  The lattice can be embedded into any 2-dimensional orientable surface, such as a torus.  The Hamiltonian for the model takes the form
\begin{equation}\label{E:targetHam}
	{H}_{\rm QD} = -\sum_v A(v) - \sum_p B(p) \, ,
\end{equation}
where $v$ denotes a vertex of the lattice and $p$ denotes a plaquette.  In the simple case of the toric code, the vertex and plaquette operators are defined by
\begin{equation}\label{E:toricAB}
	A(v) \equiv \bigotimes_{e\in+(v)} X_e \quad\quad \mbox{ and } \quad\quad
	B(p) \equiv \bigotimes_{e\in\Box(p)} Z_e \quad ,
\end{equation}
where the form for the star of a vertex $+(v)$ and the boundary of a plaquette $\Box(p)$ can be seen in Figure~\ref{F:ToricOps}, and $X_e$ and $Z_e$ are Pauli matrices acting on the qubit located on the edge $e$.  Since we are working on a square lattice, each of these terms are clearly 4-body.  For more complicated quantum double models, the $A(v)$ and $B(p)$ operators take slightly different forms, but will always consist of operators acting on the star of a vertex or the boundary of a plaquette. The exact details of the quantum double construction are given in \Sref{secprelimqd}.

When the surface in which the lattice is embedded has genus $g$, then the ground space of the toric code is $4^g$-fold degenerate, and thus can encode $2 g$ qubits. This particular encoding has generated so much interest because it is in some ways naturally robust to local errors~\cite{KitaevTC97, Dennis2002, Bravyi2010, Bravyi2010a}. The dimension of the codespace for a non-Abelian quantum double model is more complicated, but on a torus the the degeneracy of the ground space is equal to the total number of particle types~\cite{Nayak2008,Aguado2010}.

\begin{figure}
\centering
\beginpgfgraphicnamed{Figures/toriclattice}
\begin{tikzpicture}[scale=1.75]
	\draw[step=1cm, thick] (-.3,-.3) grid (3.3,3.3);
	\draw[step=1cm, ultra thick, color=blue] (0,0) grid (1,1);
	\draw[step=1cm, ultra thick, color=red] (1.5,1.5) grid (2.5,2.5);
	\foreach \x in {0,1,2} \foreach \y in {0,1,2,3}
		\filldraw[xshift=\x cm+.5cm,yshift=\y cm,color=white,draw=black]
		(0,0) circle (1.5mm);
	\foreach \x in {0,1,2,3} \foreach \y in {0,1,2}
		\filldraw[xshift=\x cm,yshift=\y cm+.5cm,color=white,draw=black]
		(0,0) circle (1.5mm);
	\foreach \a in {0,90,180,270}{
		\draw[ultra thick,color=blue,xshift=.5cm,yshift=.5cm] (\a:.5cm) circle (1.5mm);
		\draw[xshift=.5cm,yshift=.5cm] (\a:.5cm) node{{\small $Z$}};
		}
	\foreach \a in {0,90,180,270}{
		\draw[ultra thick,color=red,xshift=2cm,yshift=2cm] (\a:.5cm) circle (1.5mm);
		\draw[xshift=2cm,yshift=2cm] (\a:.5cm) node{{\small $X$}};
		}
	\draw (.5,.5) node{$p$};
	\draw (2,2) node[anchor=south east] {$v$};
\end{tikzpicture}
\endpgfgraphicnamed
\caption{The toric code defined on a square lattice with qubits on the edges. Each colored region represents one of the two types of terms in the Hamiltonian. The star terms (in red) act around a vertex $v$ with a Pauli $X$ on each qubit and the plaquette terms (in blue) act on the qubits around the boundary of $p$ with a Pauli $Z$ operator.}\label{F:ToricOps}
\end{figure}
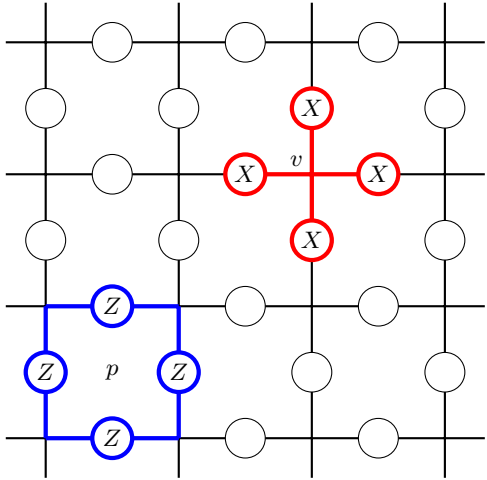

\subsection{Methods}

\subsubsection{Overview of Our Construction}

Our main result is the construction of a wholly 2-body Hamiltonian that reproduces the quantum double Hamiltonian of \Eref{E:targetHam} as its low-energy limit.  Our procedure uses only qudits of the same dimension as in the target model.   Furthermore, the couplings that we use have the same form as in the original model.  As an example, notice how each $A(v)$ and $B(p)$ term from \Eref{E:toricAB} consists of the tensor product of four Pauli $X$ terms or 4 Pauli $Z$ terms, respectively.  Our construction for the toric code will involve products of only two Pauli $X$ terms or $Z$ terms.

Each qudit on the edges of the original model is encoded in four physical qudits, shown in \Fref{F:toricPEPS}. The 2-body interactions among these four qudits will give us an effective single-qudit degree of freedom in the low-energy limit. This gives a 4-fold increase in the number of qudits required to construct our model as compared with the target model. We will couple neighbouring encoded qudits perturbatively, and the perturbation expansion will yield the desired Hamiltonian at $4^{\mathrm{th}}$ order.  This order is related to the coordination number of the underlying lattice and the number of edges bordering each plaquette, both of which are four for a square lattice. In contrast, on a honeycomb lattice terms will arise at $3^{\mathrm{rd}}$ and $6^{\mathrm{th}}$ order in perturbation theory (though there would still only be a 4-fold increase in the number of physical qudits required).


We achieve these results with a blend of techniques from condensed matter physics and quantum information theory.  Before sketching how we use these techniques, we briefly mention each one in turn.

\subsubsection{Projected Entangled Pair States.}

Our construction is inspired by {\it projected entangled pair states\/} (PEPS), a class of quantum states particularly well suited for describing the ground states of interacting quantum many-body systems~\cite{Affleck1988, Fannes1992, Hastings2006, Perez-Garcia2007a, Verstraete2006, Vidal2003a}.  Indeed, for the case of 1-dimensional systems this is provably the case~\cite{Hastings2007a}.   The basic idea of these states is to use virtual pairs of entangled systems to simulate correlations. For every coupling between neighboring systems on a lattice, a maximally entangled state (of some chosen dimension) is introduced between the systems. These virtual entangled pairs are then projected down to a ``physical'' subspace with a dimension equal to the that of the spins in the original model. An illustrative construction is depicted in Figure~\ref{F:toricPEPS} for the special case of a square lattice.

One would expect the kinds of models we are studying to have an efficient PEPS representations of their ground spaces because of the facts that they obey an area law, possess a spectral gap above the ground state, and contain finite correlations. In fact, the PEPS representation of the toric code ground space has been studied by Verstraete {\it et al.\/}~\cite{Verstraete-2006}, and a PEPS representation of general quantum double models is also known~\cite{Schuch2010}. PEPS descriptions have also been developed for all of the string-net ground spaces~\cite{Gu-09, Buerschaper-09} and the symmetries of these PEPS descriptions have been explored~\cite{Schuch2010}. The quantum double models on trivalent lattices can be mapped to string net models~\cite{Buerschaper2009}, and generalisations of the quantum double models are also being interpreted as extended string net models~\cite{Buerschaper2010}. In our construction we have an implicit PEPS representation for the quantum double models which is presumably equivalent to those already known (this is provably the case for our construction of the toric code).

\begin{figure}
\centering
\beginpgfgraphicnamed{Figures/toricPEPS}
\begin{tikzpicture}
	[a/.style={circle, fill=green!85!blue, draw=green!70!black, scale=.4},
	b/.style={fill=none, decorate, 
	decoration={snake,amplitude=.6mm,segment length=2mm, 
	pre length=1.8, post length=.8}, green!85!blue},scale=2.2,
	c/.style={fill=blue,fill opacity = .2}]

	\def\c{.01}
	\clip (\c,\c) rectangle (3-\c,3-\c);
	\draw[color=gray] (\c,\c) rectangle (3-\c,3-\c);

	\draw[step=1cm, thick] (0,0) grid (3,3);
	\foreach \x in {0,1,2} \foreach \y in {0,1,2,3}{
		\filldraw[c, shift={(\x+.5,\y)}] (0,0) circle (2mm);
		\filldraw[c, shift={(\y,\x+.5)}] (0,0) circle (2mm);
	}

	\def\tick{.075}
	\foreach \x in {0,1,2}{ 
	\foreach \y in {0,1,2}{
	\foreach \a in {0,90,180,270}{
		\filldraw[b, rotate around={\a:(\x+.5,\y+.5)}] 
		(\x+.5+\tick,\y+\tick) node [a] {} -- (\x+1-\tick,\y+.5-\tick) node [a] {};
	}}}
\end{tikzpicture}
\endpgfgraphicnamed
\caption{PEPS description on a square lattice.  Each qubit (for the toric code, or qudit in general) on an edge in the original model is replaced by four qubits (or qudits in general).  Qubits connected by a wavy line are in a maximally entangled state. Each blue circle represents a projection down to a single encoded qubit.  The quantum states in the support of these projectors are encoded qubits, entangled with each other.  The global state space contains exactly the states in the ground space of the original toric code model.  With toroidal boundary conditions, the ground space is four-fold degenerate.  Our construction proceeds by simulating these local projections and the entangling interactions with 2-body Hamiltonians.}
\label{F:toricPEPS}
\end{figure}
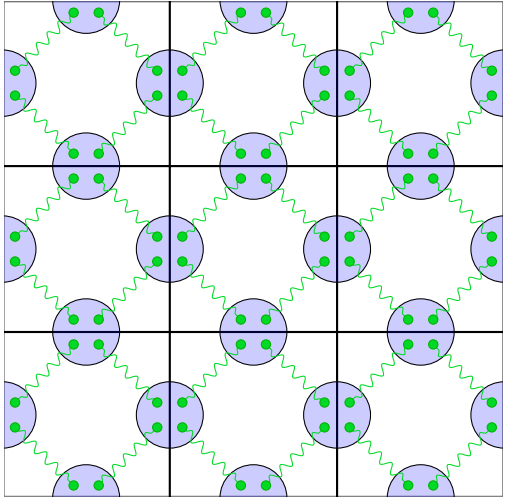

\subsubsection{Perturbation Gadgets.}

The next technique that we use is that of {\it perturbation gadgets\/}~\cite{Kempe2006, Oliveira2008, Bravyi2008a, Jordan2008}.  Perturbation gadgets are a method for systematically reducing the complexity of a many-body coupling between a large number of quantum systems.  
The gadgets generally consist of introducing some ancilla qudits which act as {\it conductors\/}, in the sense of the conductor of an orchestra.  By having a strong coupling to the conductor, $n$ separate primary qudits can synchronize their behavior in a way which mimics an $n$-body coupling at low energies, but by using couplings having only a fraction of the ``body-ness''.  By recursively applying these general constructions, one can arrive at a strictly 2-body Hamiltonian for an arbitrary $n$-body coupling.  The cost is that this coupling only occurs at a higher order in perturbation theory.

Our construction is a new variant on the perturbation gadgets approach. The concept is very similar, but instead of beginning from the original lattice and adding ancilla qubits to break up many-body interactions, we begin by encoding the qudits of the original lattice into a 4-qudit system. These systems are then coupled via relatively weak 2-body interactions, which enable us to treat the entire model perturbatively and show that it reproduces the target model in the low-energy limit. We will refer to these 4-qudit encoded systems as ``code gadgets''. Apart from the intrinsic interest of a new approach, we also manage to bypass a number of pitfalls that naive application of the perturbation gadgets can encounter. In particular, the resource cost of perturbative gadget schemes scales poorly with the system complexity, and a naive application of the technique can lead to the energy gap scaling with the system size or the fidelity of the topologically ordered states~\cite{VandenNest-08,Oliveira2008}. While this can be avoided, it remains a problem with applying the method in general. Additionally, while the couplings can be reduced to only 2-body, the nature of these 2-body couplings is in general vastly different from the couplings of the original model.  It is plausible that by taking advantage of structure a much simpler construction could be devised that is specifically tailored to the model, using couplings that exploit this structure.  This is the approach taken by Koenig~\cite{Koenig2010}, who showed that a simple ``clock'' gadget could reduce the complexity of the quantum double models to only 3-body terms. Our construction also follows a similar strategy in some sense, but uses perturbative couplings that are direct analogues of the original models' terms on simple surfaces.



The type of gadgets that we use are adapted for use with states that are ground states of local Hamiltonians and have a simple PEPS description~\cite{Bartlett-2006}.  The virtual entangled pairs of the PEPS description are promoted into real physical systems with a coupling such that their ground state is a maximally entangled state.  Then the PEPS projection can be done by using strong interactions between the systems within a code gadget which energetically favors the subspace defined by the PEPS projection.  When this coupling within a site is much stronger than the entangling coupling between sites, the resulting Hamiltonian at low energy approximates the desired many-body Hamiltonian, with the same ground space up to perturbative corrections.  This technique was originally used~\cite{Bartlett-2006, Griffin-08} to find 2-body Hamiltonians whose ground state encodes the cluster state~\cite{Raussendorf2001}, a state which is universal for measurement-based quantum computing. Our technique is very similar, in that the target models are reproduced in an encoded manner, i.e. the 4-qudit code gadgets of our construction serve as the logical spins of the target model.

\subsubsection{Subsystem Quantum Error Detecting Codes.}

Finally, we make use of quantum error-detecting codes~\cite{Shor1995, Steane1996, Gottesman1997} to ensure that all of the undesirable terms in the perturbative expansion do not couple to the low-energy sector of our model.  Recall that a quantum error-detecting code consists of a subspace of a larger Hilbert space, which is protected from some set of errors on the large Hilbert space. This protected subspace, the codespace, is used to store encoded logical information. The mapping between the codespace and the physical Hilbert space defines the encoded logical operators.  Detectable errors move the system out of the codespace, and so can be detected by a suitable measurement. 

Specifically, we use a particular type of quantum code known as a subsystem code~\cite{Kribs2005, Kribs2006}.  Compared to stabilizer codes, subsystem codes are no more powerful in terms of the number of errors they can detect, rather, their power lies in a simplification of the recovery operations required to fix the errors. The physical Hilbert space is partitioned into distinct subsystems: the logical subsystem and a ``gauge'' subsystem. As in a stabilizer code~\cite{Gottesman1997}, stabilizer operators are defined such that the logical codespace is in the mutual $+1$ eigenspace of these operators. If an error moves the system out of this $+1$ eigenspace, it can be detected by measurement of the stabilizer, and must actively be corrected. In contrast, the logical state is taken to be invariant under a transformation on the gauge system, and so any errors that occur on this gauge space are passively avoided.

The code we use on our gadget construction is designed so that any operations contributing to unwanted terms in the perturbative expansion will be detected as errors, and map the ground space to another energy eigenspace with higher energy. In this way, the code ensures that the low energy behavior of our system will remain error free, in the sense that only desirable terms (i.e. those of the target model) will remain. Errors which move the system out of the codespace need not explicitly be corrected because they will be suppressed heavily by an energy penalty for doing so. Note that the ``errors'' that occur in the gadget code are unrelated to errors appearing in the target model we are trying to replicate. Instead, they mix the protected low-energy sector with higher energy (unprotected) sectors, thus preventing our system from mimicking the target model perfectly. 


Our models are similar in many respects to the topological subsystem codes of Bombin~\cite{Bombin2010}.  These models, based on the color codes~\cite{Bombin2009,Kargarian2010} and their qudit generalizations~\cite{Sarvepalli2010}, yield a subsystem code using local 2-body gauge generators.  We note that the construction of \cite{Bombin2010} requires a 3-colorable lattice, whereas our method applies to any lattice for which the desired model possesses a PEPS description. Additionally, Bombin's models possess an exactly degenerate ground space while our models' ground spaces are only approximately degenerate (as in the Kitaev honeycomb model~\cite{Kitaev-06}), but where the degeneracy is only broken at very high order (roughly the linear size of the system).  One consequence of this exact degeneracy in the Bombin model is that it is straightforward to define quantum error correction in the ground space~\cite{Suchara2010} where perfect recovery operations will yield perfect recovery from correctable errors.  By contrast, in our model even perfect recovery operations will yield some small error due to the splitting of the ground space.  While the small error in recovery seems unavoidable in our model, it might still be the case that our model yields higher thresholds due to the substantially simpler stabilizer measurements that are required.  It is an open problem to define error correction and exactly quantify the errors incurred by the splitting for our model.  Aside from both of our constructions being viewed as a generalization of the honeycomb model, we are not aware of any deeper connection between them.  

\subsection{Discussion}

Our results allow us to replicate in the low-energy limit the Hamiltonians of certain topological models, but a more ambitious goal is to reproduce the topological order in the ground state wavefunctions of these models, as well as in their low-lying excited states.  To demonstrate that our constructions reproduce \emph{all} the topological properties of the original models, we would need to show additional properties.  First, we would need to show that these topological orders are stable under the kinds of perturbative corrections that our procedure introduces. There are two separate types of corrections which could threaten the stability of the topological properties of our models. The first are the very high-order corrections where perturbation terms can form non-trivial homology cycles on the surface of the model.  These will occur at order $2L$ in perturbation theory (with $L$ the smallest linear dimension of the surface) and split the ground space degeneracy of the encoded model. They can be heavily suppressed by increasing $L$ or by increasing the bare energy gap of the system. The second kinds of corrections to consider are those which leave the protected ground space of our code gadgets. These allow transitions into higher energy subspaces where our encodings fail. These kinds of corrections are suppressed energetically by the energy gap.  Thus, we would also need to demonstrate stability of the energy gap in the thermodynamic limit.  It would also be quite interesting to compute the topological entanglement entropy in the ground state~\cite{Hamma2005,Kitaev2006a,Levin2006,Flammia2009b}. Like the toric code model in two dimensions~\cite{Dennis2002,Nussinov2008,Castelnovo2007}, the topological order of Kitaev's non-Abelian quantum double model is not expected to persist at finite temperature. We expect that our models will have similar behaviour to the original quantum double model at finite temperature, and in particular, that they will suffer from the same ``thermal fragility'' of the topological order.

Even more ambitiously, one might hope to show that our models remain gapped even in the presence of arbitrary local perturbations, with only very small splitting of the ground state degeneracy and the degeneracy of the excited states.  For topological models whose Hamiltonians consist of a sum of commuting projectors, just such a result was shown by Bravyi, Hastings and Michalakis~\cite{Bravyi2010, Bravyi2010a}.  Unfortunately, their techniques cannot be directly applied to our models since our Hamiltonians are not sums of commuting projectors.  In fact, our models do not have frustration-free ground states either (meaning the ground states are not minimum-energy eigenstates of each term separately in the Hamiltonian) and hence other results on frustration-free systems also do not apply.

As one might expect from a perturbative construction, our effective Hamiltonian can be thought of as the target Hamiltonian plus perturbative corrections. Generically in this type of construction, the symmetries of the target model will be recovered approximately, up to these corrections. In our model, the encoded string operators of the target models will not commute with the perturbative terms in the Hamiltonian, and the corresponding ground space degeneracy will also be split. This splitting will be exponentially suppressed by the size of the lattice (as noted above), and so these symmetries will be recovered approximately as one might expect. In the quantum double models, there are also an extensive number of vertex and plaquette operators which commute with the Hamiltonian and form a quantum double algebra. It is possible to construct encoded counterparts of these operators on our model which also commute with our full Hamiltonian and form an equivalent algebra. In contrast to a generic perturbative construction, this large symmetry group is reproduced exactly by our model. This symmetry group severely constrains the arbitrarily high order terms arising from the perturbation expansion, and prevents undesirable terms from appearing. In fact, higher order terms in our effective Hamiltonian will act on the logical codespace as products of (commuting) lower order terms until the perturbative order is sufficiently high to form non-trivial loops over the lattice and break the ground space degeneracy.

As well as capturing symmetries of the target models, our construction is also natural in the sense that it is built from miniature quantum double models overlaid on each other. In the case of Abelian quantum double models (including the toric code), there is an exact correspondence between our constructions and quantum double models on simple surfaces. For these models, the ground space of our code gadget is chosen to coincide with a subspace of the quantum double ground space on a 4-qudit torus (see \Sref{subsecqdham}). That is, our codespace is stabilized by the same operators as the ground space of the quantum double Hamiltonian on this torus. These stabilizers would normally be 4-body, so in order for our codespace to obtain this property from a 2-body Hamiltonian we must sacrifice the degeneracy of one of the two qudits a torus can typically encode. In addition, the perturbative bond terms we introduce can be interpreted as quantum double models on a 2-qudit sphere (see \Sref{subsecqdcouple}). All the terms present in our model reflect the construction of the quantum double models on small surfaces. In non-Abelian models, the correspondence in the Hamiltonian is not as precise due to the distinction between left and right regular representations. However it will still be true that the ground spaces of our code gadgets will correspond to a subspace of the quantum double ground space on a 4-qudit torus, and the ground space of the perturbative bond terms will correspond to the ground space of the relevant quantum double model on a 2-qudit sphere. 


Our work is inspired by the construction of Bartlett and Rudolph~\cite{Bartlett-2006}, who used encoded qubits to reproduce the cluster state as the ground state of a 2-body Hamiltonian. As with the model we present, the encoding is closely related to the PEPS description of the target state. Our work generalises this type of construction by using a subsystem code (as opposed to a subspace code). In order to reproduce the PEPS space as the ground space of a 2-body Hamiltonian, we have had to sacrifice the extra gauge degrees of freedom in our model. This would not have been possible if we had used a subspace code, where these gauge degrees of freedom are not available.

The procedure we present reproduces the quantum double Hamiltonian in the perturbative coupling limit. If we consider the opposite (strong coupling) limit in our model, the system will act as a disconnected set of maximally entangled pairs. In this limit, the lattice can be thought of as a set of disconnected quantum double models on spheres (up to the caveats noted above). There must be some phase transition(s) between these states, and the respective topologies in the two limits are suggestive of the kind of behavior studied by Gils {\it et al.\/}~\cite{Gils2009}. The phase transition would be between different topologies in the sense that the quantum double model would act on the entire lattice together in one limit, and at some critical coupling strength break down to act on disconnected portions of it.

A robust topologically ordered system would also have anyonic low-energy excitations which, ideally, would be similar to those of the desired quantum double model.  One would need to consider the effect of the perturbative corrections in our model on these anyons and any other excitations, as for example in~\cite{Dusuel2008, Schmidt2008, Vidal2008}.

It would also be interesting to explicitly and rigorously show that cooling our Hamiltonians by coupling to a local bath can bring them to the ground space quickly.  (Cooling to a particular ground \emph{state} of the degenerate ground space is more difficult, but also quite interesting.)  Because of the frustration in the model, it isn't immediately obvious that this can be achieved efficiently, i.e.\ in an amount of time polynomial in the size of the system.  But because the low-energy effective theory is frustration-free, it certainly seems plausible that the cooling can be done efficiently.

It is noteworthy that the codes we utilise in our construction for general quantum double models are examples of extensions of the stabilizer formalism to non-Abelian groups. This has not been closely studied previously, and in that sense the quantum codes we use may be of interest in their own right.


We believe that the kind of approach we employ here to reproduce topological orders may be more generally applicable to other systems with efficient PEPS representations. Some work towards extending this treatment to the class of string net models~\cite{LevinSN05} supports this belief, with some caveats, and will be presented in a future publication.

Finally, we thank Miguel Aguado, Sergey Bravyi, Robert Koenig, Spyridon Michalakis, David Poulin, and Guifre Vidal for interesting and fruitful discussions.  SDB and ACD acknowledge support from the ARC via the Centre of Excellence in Engineered Quantum Systems (EQuS), project number CE110001013. 
Research at Perimeter Institute is supported by the Government of Canada through Industry Canada and by the Province of Ontario through the Ministry of Research and
Innovation.  
This work has been supported in part by NSA/ARO under Grant No.\ \mbox{W911NF-09-1-0442} and by NSF under Grant No.\ \mbox{PHY-0803371}.

%
%

\section{Example: The Toric Code}

As a simple and illustrative example of our scheme, we now demonstrate how to construct a two-body Hamiltonian for which the low-energy behavior reproduces the standard toric code Hamiltonian on a square lattice.  This simple example possesses all of the key features of our general construction for the quantum double models.

\subsection{PEPS Representation of the Toric Code}

We begin our construction by replacing each qubit on the edges of the toric code lattice with four qubits, as in Fig.~\ref{F:toricPEPS}.  We use the term {\it bond\/} to refer to the wavy lines in Fig.~\ref{F:toricPEPS} that connect the maximally entangled pairs of physical qubits.  By contrast, we use the word {\it edge\/} to denote the edges of the original lattice.  In the PEPS description, the projection operators acting on each edge will entangle these four qubits into a single qubit.  We achieve this projection in the ground space of a Hamiltonian defined on each edge $e$.  
\begin{align}\label{E:localH}
        H(e) = - \siteT{X}{I}{I}{X} - \siteT{I}{X}{X}{I} - \siteT{Z}{Z}{I}{I} - \siteT{I}{I}{Z}{Z} \, .
\end{align}

This notation is a convenient visual shorthand for the tensor product of the operators acting on the given physical qubits. For each edge, the collection of 4 qubits forms our ``code gadget''. The Hamiltonian contains only two-body terms acting within the gadget itself.

Instead of the explicit projection mechanism of the PEPS scheme to reduce the Hilbert space, our model simply supresses by energy penalty states which lie outside the desired PEPS projection. It can be shown that the projectors to the ground space of our edge Hamiltonian are equivalent to the projectors in Ref.~\cite{Verstraete-2006} since they are equal modulo the gauge freedom in choosing the PEPS description; however, explicitly demonstrating this equivalence is tedious (though straightforward) so we omit this.


Next, we introduce entanglement across the bonds by coupling sites on different edges.  Thus for each bond $b$, define the perturbation term
\begin{align}\label{E:localV}
        V(b) = - \bond{X}{X} - \bond{Z}{Z} \, ,
\end{align}
chosen because it possesses a maximally entangled state as its ground state\footnote{An alternative choice would be $V(b) = - \bond{Y}{Y}$. This choice also approximately realizes the toric code Hamiltonian in the small $\lambda$ limit. The resulting Hamiltonian \Eref{E:toricfullHamil} is precisely equivalent to Kitaev's exactly solvable honeycomb model on a mosaic tiling~\cite{Yang2007}. We do not consider this possibility further since it is unclear how to generalize it to more complicated quantum double models. It is interesting to note, however, that this alternative choice results in a model with topological order in the large $\lambda$ limit, in contrast to our models which become valence bond solids.}.
Our unpertubed Hamiltonian is summed over all edges $e$:
\begin{equation}
H_0 = \sum_e H(e) \,,
\end{equation}
and our perturbation term is summed over all bonds $b$:
\begin{equation}
V =  \sum_b V(b)\,.
\end{equation}

For those readers familiar with PEPS, it may seem counterintuitive to treat the bond term as small compared to the code gadget Hamiltonian (which is simulating the PEPS projection). We will see that this is in fact the correct approach to recover the target model. We introduce a coupling strength $\lambda$ which is a small parameter compared to the strength of the main terms in our Hamiltonian (which we have taken to have unit norm). The full Hamiltonian describing our lattice is then given by
\begin{align}\label{E:toricfullHamil}   {H} = H_0 + \lambda  V \, .\end{align}

Now we need to compute the perturbative low-energy effective Hamiltonian to leading nontrivial order in $\lambda$.  We will find the exact ground space of $H(e)$ in the next section, and then show that the perturbations $\lambda V$ will generate operators which reproduce an encoded toric code Hamiltonian (Eq.~\ref{E:targetHam}-\ref{E:toricAB}) at fourth order in $\lambda$.

\subsection{Solving the Code Gadget Hamiltonians}

We must first demonstrate that $H(e)$ of an edge $e$ has a two-dimensional degenerate ground space. We will show that this ground space is in fact the codespace of a subsystem quantum error-detecting code, which will greatly assist the perturbative analysis in the next section. Our analysis of this Hamiltonian for the toric code construction follows Bacon~\cite{BaconThesis}. Because we are always working on a particular arbitrary edge (within a particular code gadget), we will suppress the label $e$ in this section. 

The code gadget Hamiltonian $H$ (Eq.~\ref{E:localH}) posesses a number of constants of motion.  That is, we can define operators that commute with each other, and with the Hamiltonian.  In fact some of these operators are the {\it stabilizers\/} of a quantum code, so we label them $S$.  They form a commutative group, and are generated by
\begin{align}\label{eqsx}
        S_X \equiv \siteTnoe{X}{X}{X}{X} \quad \mbox{ and } \quad
        S_Z \equiv \siteTnoe{Z}{Z}{Z}{Z} \, .
\end{align}

We can also define other joint operators to complete the algebra of our code gadget.  We will call these operators {\it gauge\/} operators and {\it logical\/} operators, and denote them with appropriate subscripts.
\begin{align}\label{eqxl}
        X_\subfont{G} \equiv \siteTnoe{I}{X}{X}{I} \, , \quad
        Z_\subfont{G} \equiv \siteTnoe{Z}{Z}{I}{I} \, , \quad \mbox{ and } \quad
        X_\subfont{L} \equiv \siteTnoe{X}{X}{I}{I} \, ,\quad
        Z_\subfont{L} \equiv \siteTnoe{I}{Z}{Z}{I} \, .
\end{align}

We see immediately that these operators encode two orthogonal copies of the Pauli algebra, so they define two encoded qubits. In terms of these new operators, we can rewrite the gadget Hamiltonian
\begin{align}\label{eqTCh0}     
        H(e) = -X_\subfont{G}(1+S_X)-Z_\subfont{G}(1+S_Z) \, .
\end{align}

The protected subspace of our code, also corresponding to the ground space of the Hamiltonian, is a subspace of the $+1$ eigenspace of the stabilizers $S_X$ and $S_Z$. Any single qubit error will anticommute with at least one of these, and so could be detected (though not unambiguously) by measurement of the stabilizers. This means that any single qubit error will neccesarily move the system out of its ground space, and will be suppressed by the energy penalty for doing so.


We can easily check that the logical operators $X_\subfont{L}$ and $Z_\subfont{L}$ also commute with $H$.  However, neither $X_\subfont{G}$ nor $Z_\subfont{G}$ commutes with $H$.  Given all these facts, $H$ decomposes into a direct sum of four copies of $\subfont{L} \otimes \subfont{G}$, each labeled by the pairs of eigenvalues $(\pm 1, \pm 1)$ of $S_X$ and $S_Z$, and furthermore the energies in the logical space are degenerate.  Thus, we only have to solve the Hamiltonian on the gauge subspace of each stabilizer eigenvalue to find the ground space.  

It turns out the ground space is contained in the $(+1,+1)$ block, and it is exactly two-fold degenerate.  We provide a proof of this statement and its generalization to all the quantum double models in \Aref{provethmnags}. It can be shown that within this codespace the encoded computational basis states take the (unnormalised) form:
\begin{align}
\ket{0_\subfont{L}} &= \left(1+\sqrt{2}\right)\left[\siteTnoeB{\ket{0}}{\ket{0}}{\ket{0}}{\ket{0}} + \siteTnoeB{\ket{1}}{\ket{1}}{\ket{1}}{\ket{1}}\right] + \siteTnoeB{\ket{0}}{\ket{1}}{\ket{1}}{\ket{0}} + \siteTnoeB{\ket{1}}{\ket{0}}{\ket{0}}{\ket{1}} \\
\ket{1_\subfont{L}} &=  \left(1+\sqrt{2}\right)\left[\siteTnoeB{\ket{1}}{\ket{1}}{\ket{0}}{\ket{0}} + \siteTnoeB{\ket{0}}{\ket{0}}{\ket{1}}{\ket{1}}\right] + \siteTnoeB{\ket{1}}{\ket{0}}{\ket{1}}{\ket{0}} + \siteTnoeB{\ket{0}}{\ket{1}}{\ket{0}}{\ket{1}}
\end{align}

Now that we have determined that the ground space encodes a qubit in an error-detecting code, we can perform the perturbative analysis to compute the low-energy effective Hamiltonian.

\subsection{Perturbation Analysis}\label{S:perttoric}

We now introduce the perturbative coupling of Eq.~\ref{E:localV} between our encoded qubits on the lattice.  We use the Green's function perturbation method, following Kitaev~\cite{Kitaev-06} (see also \cite{Bergman2007}) to calculate the leading non-trivial order in the effective Hamiltonian, defined as $H_{\mathrm{eff}} = E_0+\Sigma(E_0)$, with the unperturbed ground state energy of the lattice $E_0$ and $\Sigma$ the self-energy. In this, we have approximated the self energy as being independent of $E$, for $E\approx E_0$. More details of our perturbation formalism are in Appendix~\ref{secpertform}.

The key to our analysis is the fact that quantum codes map detectable errors to orthogonal states.  In terms of the gadget Hamiltonians $H(e)$, this means that any single-qubit Pauli operator anti-commutes with either (or both) of $S_X$ or $S_Z$, and hence maps ground states to orthogonal states, since these states must lie in some block of $H(e)$ other than the $(+1,+1)$ block.

The consequence is that the perturbation analysis greatly simplifies.  It immediately gives the result that all odd-order perturbation terms will vanish, as they will necessarily leave two code gadgets in excited states.  The terms at second order only contribute an energy shift, since those that don't vanish act twice on the same qubits, hence they act proportionally to the identity.  The first non-trivial terms appear at fourth order, and we can write the 4th-order effective Hamiltonian as follows:
\begin{equation}\label{E:toric4eff}
        H_{\mathrm{eff}}^{(4)} = \lambda^4\Upsilon V \bigr(G_0(E_0)V\bigr)^3 \Upsilon \, ,
\end{equation}
where $\Upsilon$ is the projector to the ground space of the unperturbed system and $G_0(E)=(E-H_0)^{-1}(1-\Upsilon)$ is the Green's function (resolvent) projected to vanish on ground states. The $\Upsilon$ will project the stabilizer and gauge degrees of freedom down to a single state, but will act identically on the logical degree of freedom.

The non-trivial fourth order terms arise by constructing joint operators around a plaquette or around a vertex which leave all the code gadgets in the ground space. By expanding Eq.~\ref{E:toric4eff} and ignoring constant energy shifts, we can express the effective Hamiltonian to $4^{\mathrm{th}}$ order as
\begin{equation}
         H_{\mathrm{eff}} \propto - \lambda^4 \sum_v \Upsilon \hat{A}(v)\Upsilon
          - \lambda^4 \sum_p \Upsilon\hat{B}(p)\Upsilon \, ,
\end{equation}
where $v$ and $p$ sum over all vertices and plaquettes, respectively.  The operators $\hat{A}(v)$ and $\hat{B}(p)$ are each a sum of two terms,
\begin{align}   \hat{A}(v) = \hat{A}(X,v) + \hat{A}(Z,v) \quad \mbox{ and } \quad
        \hat{B}(p) = \hat{B}(X,p) + \hat{B}(Z,p) \, ,\end{align}where these terms are quite cumbersome to express algebraically, so we define them pictorially in Figure~\ref{F:ToricEncodedOps}.  Basically, these operators act on the pairs of qubits in an edge that are nearest to the center of a given plaquette or vertex.  

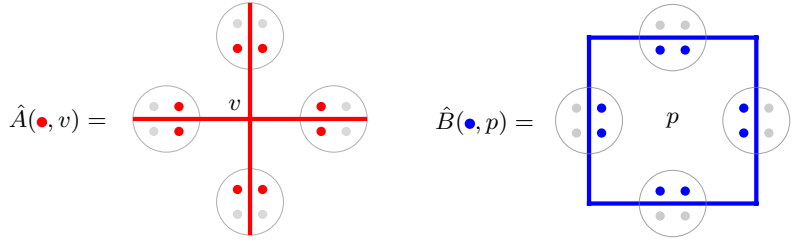
\begin{figure}
\centering
\beginpgfgraphicnamed{Figures/encodedAtoric}
\begin{tikzpicture}
        [a/.style={circle, fill=red, scale=.4},
        b/.style={circle, fill=gray, opacity=.3, scale=.4},
        scale=2.2]

        \node at (.5,.5) [anchor = south east] {$v$};
        \node at (-.6,.5) {$\hat{A}(\tikz[scale=.5] \fill[red] (1ex,1ex) circle (1ex);,v) = \ \ $};

        \foreach \a in {0,90,180,270}
                \draw[gray,opacity=.6,rotate around={\a:(.5,.5)}] (0,.5) circle (2mm);

        \def\g{.3}
        \draw[red,ultra thick,shift={(.5,.5)}] (-1+\g,-1+\g) grid (1-\g,1-\g);

        \def\tick{.075}
        \foreach \a in {0,90,180,270}{
                \draw[rotate around={\a:(.5,.5)}] 
                        (.5+\tick,\tick) node [a] {} (.5-\tick,\tick) node [a] {};
                \draw[rotate around={\a:(.5,.5)}] 
                        (.5+\tick,-\tick) node [b] {} (.5-\tick,-\tick) node [b] {};
        }
\end{tikzpicture}
\endpgfgraphicnamed
\quad \quad
\beginpgfgraphicnamed{Figures/encodedBtoric}
\begin{tikzpicture}
        [a/.style={circle, fill=blue, scale=.4},
        b/.style={circle, fill=gray, opacity=.4, scale=.4},
        scale=2.2]

        \def\g{.013}
        \draw[blue,ultra thick] (-\g,-\g) grid (1+\g,1+\g);

        \node at (.5,.5) {$p$};
        \node at (-.6,.5) {$\hat{B}(\tikz[scale=.5] \fill[blue] (1ex,1ex) circle (1ex);,p) = \ $};

        \foreach \a in {0,90,180,270}
                \draw[gray,opacity=.7,rotate around={\a:(.5,.5)}] (0,.5) circle (2mm);

        \def\tick{.075}
        \foreach \a in {0,90,180,270}{
                \draw[rotate around={\a:(.5,.5)}] 
                        (.5+\tick,\tick) node [a] {} (.5-\tick,\tick) node [a] {};
                \draw[rotate around={\a:(.5,.5)}] 
                        (.5+\tick,-\tick) node [b] {} (.5-\tick,-\tick) node [b] {};
        }
\end{tikzpicture}
\endpgfgraphicnamed
\caption{Encoded operators.  For each operator, the non-trivial operator acts on the colored qubits.  For example, $\hat{A}(X,v)$ is a tensor product of $X$ operators on each of the colored qubits surrounding the vertex $v$.\label{F:ToricEncodedOps}}
\end{figure}

We can derive from Fig.~\ref{F:ToricEncodedOps} that $\hat{A}(X,v)$ acts like a tensor product of logical $X_\subfont{L}$ operators on each of the edges surrounding $v$.  This is not immediately obvious; we must use the fact that in the ground space $X_\subfont{L} = S_X X_\subfont{L}$ to exchange the action of the logical operators between pairs of qubits at a particular edge.  Similarly, $\hat{B}(Z,p)$ acts like a tensor product of logical $Z_\subfont{L}$ operators around the plaquette $p$.  The other operators $\hat{A}(Z,v)$ and $\hat{B}(X,p)$ act like gauge operators in a similar fashion.  When these encoded gauge operators are mapped back to the ground space by $\Upsilon$, they contribute only a constant energy shift, which we can ignore.

We can think of these $\hat{A}$ and $\hat{B}$ operators as acting equivalently to some logical operators within the ground space. If we define
\begin{align}   \hat{A}_\subfont{L}(v)\equiv\bigotimes_{e\in +(v)}X_\subfont{L}^e\qquad \mbox{ and } \qquad \hat{B}_\subfont{L}(p) \equiv \bigotimes_{e\in\Box(p)} Z_\subfont{L}^e \, ,\end{align}
then we can see that within the codespace, they are equivalent (up to muliplicative and additive constants) to $\hat{A}$ and $\hat{B}$ respectively. That is,
\begin{align}   \Upsilon \hat{A}_\subfont{L}(v) \Upsilon \propto \Upsilon\hat{A}(v)\Upsilon-\mathrm{const} \qquad \mbox{ and } \qquad \Upsilon\hat{B}_\subfont{L}(p) \Upsilon\propto\Upsilon\hat{B}(p)\Upsilon-\mathrm{const} \, .\end{align}

We can think of the logical Hamiltonian acting within the codespace as being comprised of these operators, such that
\begin{equation}\label{E:toricEffHam}
        H_\subfont{L}  =  -\lambda^4 \sum_v\hat{A}_\subfont{L}(v) - \lambda^4 \sum_p\hat{B}_\subfont{L}(p) \, .
\end{equation}

When restricted to the codespace, this is exactly the effective Hamiltonian we previously derived (again, up to multiplicative constants and energy shifts), so that
\begin{equation}
H_{\mathrm{eff}} \propto \Upsilon H_\subfont{L}\Upsilon+\mathrm{const}
\end{equation}

Noting that the logical operators $\hat{A}_\subfont{L}$ and $\hat{B}_\subfont{L}$ act on the logical state exactly like the toric code vertex and plaquette terms, we can see that on the logical space, our effective Hamiltonian is the toric code Hamiltonian of Eq.~\ref{E:toricAB} up to constants, as claimed. 

The higher order terms in the expansion for the self-energy (see Appendix A) will generally act on the logical space like products of the terms appearing in \Eref{E:toricEffHam}. Moreover, for low energies all these terms will be negative, since the perturbation term and the Green's function will both be non-positive and so each term in the expansion will be negative. All these terms in the self-energy expansion will commute, and so the ground space should remain the $+1$ eigenspace of the terms in \Eref{E:toricEffHam} as desired. There will be some corrections to the excited spectrum of the effective Hamiltonian due to these higher order corrections and due to the energy dependence of the self-energy as discussed in Appendix A But as we are mainly interested in the topologically ordered ground space, this does not concern us especially. At very high order it will be possible to construct terms which run all the way around the torus. These errors will  corrupt the logical state. If the linear size of the torus is $N$, these terms will appear at $2N^{\mathrm{th}}$ order, and will be suppressed by a factor of $\lambda^{2N}$.

This result allows us to take a system of only two-body couplings, and in the low energy limit reproduce the Hamiltonian of the toric code. This means we might expect to be able to use the topological properties of the toric code to protect quantum information without the requirement for experimentally problematic many-body couplings. A similar result of obtaining the toric code in a limit was observed in Kitaev's honeycomb model~\cite{Kitaev-06}. In contrast to our construction, this honeycomb model is exactly solvable. Although we can only solve our model perturbatively, we can generalize it relatively easily to more complicated quantum double models (and lattices other than square), as will be seen in the following sections.

%
%

\section{Review of Quantum Double Models}\label{secprelimqd}

The quantum double models consist of coupled finite-dimensional quantum systems on the edges of a lattice, and their ground states exhibit topological order \cite{KitaevTC97}.  In this section, we will define the quantum double Hamiltonian that will become the target model of the perturbative two-body systems that we will work with in the subsequent sections.  

As with the toric code, we will work with a square lattice for concreteness.  The lattice can be embedded into any orientable 2-dimensional surface.  To each edge $e$ of the lattice we will associate an orientation, as in Figure~\ref{figdirlattice}.  Although this orientation could be arbitrary, we chose the orientation in Figure~\ref{figdirlattice} because it has the convenient feature that each vertex can be labeled with either a ``$+$'' or a ``$-$'' sign.  In fact, any bipartite lattice can be partitioned in such a way.

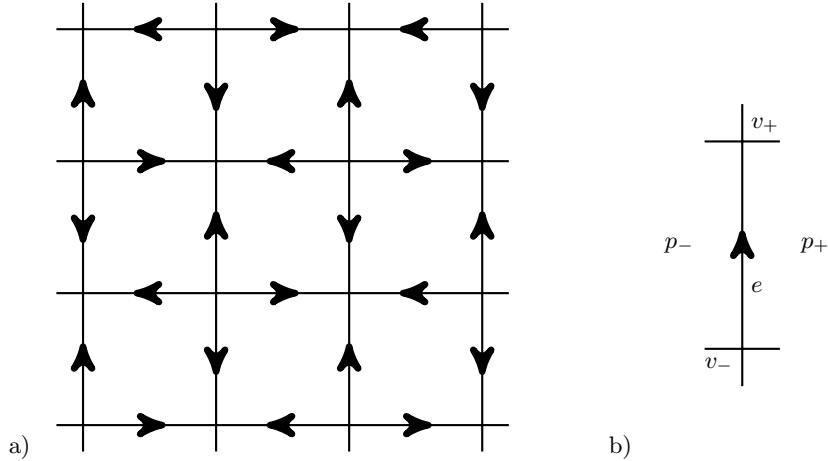
\begin{figure}
\centering
a)
\subfloat{
\beginpgfgraphicnamed{Figures/directedlatticeA}
\begin{tikzpicture}[>=stealth',line width=2.5,scale=1.75]
	\def\g{.2} 
	\draw[step=1cm,thick] (-\g,-\g) grid (3+\g,3+\g);
	\clip (-\g,-\g) rectangle (3+\g,3+\g); 

	\def\len{.15}
	\def\off{.47}
	\foreach \x in {0,2}{ \foreach \y in {0,2}{ 
	\foreach \a in {0,90,180,270}{
		\draw[->,rotate around={\a:(\x,\y)}] (\off+\x,\y)--(\off+\len+\x,\y);
		\draw[->,shift={(1,1)},rotate around={\a:(\x,\y)}] (\off+\x,\y)--(\off+\len+\x,\y);
	}}}
\end{tikzpicture}
\endpgfgraphicnamed
}
\qquad\qquad
b)
\subfloat{
\beginpgfgraphicnamed{Figures/directedlatticeB}
\begin{tikzpicture}[scale=2.75]
	\def\g{.18} 
	\draw[step=1cm,thick] (-\g,-\g) grid (0+\g,1+\g);
	
	\draw (0,1) node[anchor=south west] {$v_+$}
		(0,0) node[anchor=north east] {$v_-$}
		(-.3,.5) node {$p_-$} (.35,.5) node {$p_+$}
		(0,.3) node[anchor=west] {$e$};

	\def\off{.47}
	\draw[->,>=stealth',line width=2.5] (0,\off)--(0,\off+.1);
	
	\fill[opacity=0] (0,-.4) circle (.1); 
\end{tikzpicture} 
\endpgfgraphicnamed
}
\caption{a) A directed square lattice, and the orientation of $+$ and $-$ vertices and plaquettes relative to edge direction. Each vertex consists either of all ``inward'' edges or all ``outward'' edges.  Plaquettes consist of alternately directed edges as you traverse their boundary. b) If an edge spans the pair of vertices $(v_+,v_-)$, then the edge is oriented toward $v_+$.  The plaquettes are labeled with signs, where $p_+$ is on the right of the edge, following the given orientation.\label{figdirlattice}}
\end{figure}

We associate each quantum double model with a finite group $G$, and a local Hilbert space for each edge $\C^{|G|}$, with $|G|$ the order of the group.  On each edge $e$ of the lattice, there exists a natural orthonormal basis $\{ \ket{g}_e$, $g \in G \}$, for these degrees of freedom.  The total Hilbert space is then the tensor product of the local Hilbert spaces over all the edges.  

We now define a number of operators that act on these edge degrees of freedom.  For each edge $e$, define operators associated with left ($+$) and right ($-$) group multiplication and group projectors as follows:
\begin{align}
	L_+^g(e) &\equiv \sum_h\ketbra{gh}{h}_e \, , & L_-^g(e) &\equiv \sum_h\ketbra{hg^{-1}}{h}_e \, , \\
	T_+^g(e) &\equiv \proj{g}_e \, , &  T_-^g(e) &\equiv \proj{g^{-1}}_e \,.
\end{align}
The operators on a single edge form an algebra defined by commutation relations
\begin{equation}
	L_{\pm}^g T_{\pm}^h = T_{\pm}^{gh} L_{\pm}^g
	\quad , \quad
	L_{\pm}^g T_{\mp}^h = T_{\mp}^{hg^{-1}} L_{\pm}^g \,.
\end{equation}
Clearly, operators acting on different edges commute.  

We also associate a sign $\pm$ to each vertex and plaquette relative to their incident edge.  This is illustrated in Figure~\ref{figdirlattice}.  If an edge spans the vertices $(v_-,v_+)$, then the arrow along the edge points away from $v_-$ and toward $v_+$.  Plaquettes to the left of an edge when looking along its direction are labelled $p_-$, while those to the right are labelled $p_+$.

It is convenient to associate a particular operator with a (plaquette, edge) pair or a (vertex, edge) pair. That is, depending on the sign of the vertex or plaquette under consideration at the time ($v_\pm$ or $p_\pm$) respective to the edge under consideration, the sign of the group multiplication and projection operators can be inferred. Explicitly, we define
\begin{equation}\label{eqlvert}
	L^g(e,v_\pm) \equiv L^g_\pm(e) \quad , \quad T^g(e,p_\pm) \equiv T^g_\pm(e) \, .
\end{equation}

The $L$ operators play the role of a generalized Pauli $X$ operator in the group element basis, insofar as they move a particular state $\ket{g}$ through the space of elements. To make the analogy with the toric code more apparent, we can also construct some operators which act as generalized Pauli $Z$'s for an arbitrary group algebra. Let $\pi$ be a unitary irreducible representation of $G$.  Then we can define a type of Fourier transform by
\begin{equation}
	Z^{\pi_{ij}}_{\pm} \equiv \sum_g[\pi(g)]_{ij}T_{\pm}^g \, ,
\end{equation}
where $[\pi(g)]_{ij}$ is the $(i,j)^{\mathrm{th}}$ element of the representation matrix for group element $g$ in representation $\pi$. Equivalently, we can invert this expression to obtain
\begin{equation}
	T^g_{\pm} = \frac{1}{|G|}\sum_{\pi}d_{\pi}\sum_{ij}[\pi(g)]_{ij}Z^{\pi_{ij}}_{\pm} \, ,
\end{equation}
where the sum is over the complete set of unitarily inequivalent irreps of $G$ and $d_\pi$ is the dimension of the irrep $\pi$.  Although the quantum doubles are typically defined in terms of $T$ operators, the algebra of these generalized $Z$ operators gives the most convenient form for a particular calculation later on.

Given these preliminary operators, for each vertex $v$ we can define operators
\begin{equation}
	A^g(v) \equiv \bigotimes_{e \in +(v)} L^g(e,v) \, ,
\end{equation}
where $+(v)$ is the set of edges incident on the vertex $v$ (recall Figure~\ref{F:ToricOps}).  We can average these operators over the group to give the projector
\begin{equation}
	A(v) \equiv \frac{1}{|G|} \sum_g A^g(v) \, .
\end{equation}

Given a plaquette $p$ and a fiducial edge on its boundary that we label $e_1$, we can define an operator
\begin{equation}
	B^g(p) \equiv \sum_{g_k \cdots g_1 = g}\, \bigotimes_{e_i \in \Box(p)} T^{g_i}(e_i,p) \, ,
\end{equation}
where $e_i$ are the boundary edges taken as the plaquette is traversed clockwise starting with $e_1$, and there are $k$ total edges on the boundary of $p$.  These operators are all orthogonal projectors, but note that this definition depends on the choice of the fiducial edge $e_1$.  However, if we consider the operator
\begin{align}\label{E:BpDefinition}
	B(p) \equiv B^1(p) \, ,
\end{align}
where $1$ is the identity element of the group $G$, then it is easy to see that $B$ no longer depends on the choice of this fiducial edge.


For all $p$ and all $v$, the $B(p)$ operators and the $A(v)$ commute pairwise amongst themselves and each other.  Finally, the following Hamiltonian defines the quantum double model
\begin{equation}\label{qdtargetham}
	H = -\sum_v A(v) - \sum_p B(p) \, ,
\end{equation}
in a fashion which directly generalizes the toric code.

\subsection{Simplifications in the Case of Cyclic Groups}

It will be instructive to treat the cyclic groups $\Z_d$ before moving on to the general case.  For that reason, we revisit the above discussion specialized to this setting.  Because each of the $|G|$ representations of the (Abelian) groups $\Z_d$ are 1-dimensional, we can relate the $L$ and $Z$ operators by a simple discrete Fourier transform for these groups. In cyclic groups (with $d=|G|$), the group multiplication operation is addition (modulo $d$), and so the left and right multiplication operations are equivalent. Following this, the convention for the identity element in cyclic groups is $0$ as opposed to $1$ for general groups.
With no need for unique left and right multiplication operators, we define the cyclic $L$ operator by
\begin{equation}
	L \equiv \sum_{h=0}^{d-1}\ketbra{h+1}{h} \, .
\end{equation}

The addition within the ket is performed modulo $d$. Note that the group action of any other element can be achieved in these models by taking powers of this (unitary) operator. We also define a set of group projection operators
\begin{equation}
	T^g \equiv \proj{g} \,,
\end{equation}

Following the general case, we also define a generalized Pauli $Z$ operator. As representations of the cyclic groups are 1-dimensional, we can define a primitive $Z$ corresponding to a representation $\omega$
\begin{equation}
	Z \equiv \sum_{h=0}^{d-1}\omega^hT^h \, ,
\end{equation}
where $\omega$ is a primitive $d^{\mathrm{th}}$ root. Other representations of the group correspond to powers of $\omega$, so to obtain the operators corresponding to these representations, we need only take powers of this $Z$ operator. Thus we regard the powers of $Z$ as being labelled by representations of $\mathbb{Z}_d$.

If we take a discrete fourier transform on the basis $\left|g\right>$, we obtain:
\begin{equation}
	\ket{\gamma} = \frac{1}{\sqrt{d}}\sum_j\omega^{\gamma j}\ket{j} \, 
\end{equation}
where we label the Fourier basis with Greek letters (corresponding to irreps).  This transformation diagonalizes the $L$ operators:
\begin{equation}
	L = \sum_{\gamma}\omega^{-\gamma}\proj{\gamma} \, ,
\end{equation}
and so we can see that the $Z$ and $L$ operators are simply a basis change away from each other in these cyclic models, as was the case with the Pauli matrices for the toric code. As they are both unitary, we can also say
\begin{eqnarray}
	L^{\dagger}&=&L^{-1}=L^{d-1} \, ,\\
	Z^{\dagger}&=&Z^{-1}=Z^{d-1} \, .
\end{eqnarray}

In terms of the definition of the quantum double model for cyclic
 groups, the only changes we need make to become consistent with this simplified set of operators is to slightly redefine the associations of $L$ and $T$ operators with $\pm$ vertices and plaquettes, i.e.
\begin{equation}
	L^g(e,v_\pm) \equiv L^{\pm g}(e) \quad , \quad T^g(e,p_\pm) \equiv T^{\pm g}(e) \, .
\end{equation}

With this in mind, the quantum double Hamiltonian is defined exactly as in the general case.

%
%

\section{Our Construction for the Cyclic Quantum Double Models}\label{secabqd}


In this section we will show how our construction on the toric code generalizes naturally to the quantum doubles of cyclic groups. The toric code model corresponds to the quantum double of the group $\Z_2$; here we extend this treatment to the quantum double of $G = \Z_d$, where $|G| = d$ is the order of the group. This analysis could be extended to general Abelian groups.  However, for simplicity, and because the fully general case is considered in the next section, we restrict our attention to cyclic groups in this section.

In order to reproduce the cyclic quantum double models, two features must be added to the simple toric code construction. To begin with, qubits at each site must be replaced by $d$-dimensional qudits, with appropriate generalized Pauli operators defined on them, as introduced in \Sref{secprelimqd}. The group multiplication operator $L$ plays the role of the $X$ operator in the toric code, and the newly generalized $Z$ operator plays the role of the Pauli $Z$. 
These operators obey the commutation relation
\begin{equation}
Z^aL^b = \omega^{ab}L^bZ^a
\end{equation}

The other feature we will add to our construction at this juncture is the notion of directed edges, as discussed in \Sref{secprelimqd}. As in the toric code ($\mathbb{Z}_2$) case, we will now proceed with our construction explicitly on the square lattice.

\subsection{Code Gadgets on Lattice Edges}

We use a very similar construction to the toric code to encode our qudits for cyclic quantum double models. Each logical qudit is encoded using a subsystem code constructed from 4 physical qudits (Figure~\ref{F:toricPEPS} shows the scheme). The $d^4$-dimensional space of these edges is partitioned as
\begin{equation}
        \mathcal{H}(e) = \bigoplus_S\mathcal{H}_\subfont{L}\otimes\mathcal{H}_\subfont{G}
\end{equation}
where the direct sum is over eigenvalues of stabiliers $S$. Our codespace is in the $+1$ eigenspace of two stabilizer operators $S_L$ and $S_Z$ defined below. The remaining $d$-dimensional degrees of freedom are encoded as qudits, one of which (the gauge qudit $\mathcal{H}_\subfont{G}$) we fix in a single state in the codespace. The other qudit is used as our logical space $\mathcal{H}_\subfont{L}$. If an error occurs, it will flip a stabilizer operator, move the code gadget out of the codespace, and incur an energy penalty.

Physically, we provide the codespace with these properties as the ground space of a 2-body Hamiltonian. Before we write it, we will first introduce the gauge, logical, and stabilizer joint operators as we did for the toric code.
\begin{align}\label{eqabop1}
S_L \equiv&\siteA{L}{L}{L^{\dagger}}{L^{\dagger}} \ , & 
L_\subfont{G} \equiv&\siteA{I}{L^{\dagger}}{L}{I} \ , & 
L_\subfont{L} \equiv&\siteA{L}{L}{I}{I} \ , \nonumber \\
S_Z \equiv&\siteA{Z}{Z^{\dagger}}{Z^{\dagger}}{Z} \ , &
Z_\subfont{G} \equiv&\siteA{Z^{\dagger}}{Z}{I}{I} \ , &
Z_\subfont{L} \equiv&\siteA{I}{Z}{Z}{I} \, .
\end{align}

To avoid confusion as much as possible, we will distinguish typographically $\subfont{G}$ for gauge and $G$ for group, similarly $\subfont{L}$ for logical and $L$ for group multiplication. It is simple to verify that the operators defining each seperate degree of freedom commute with each other (i.e., stabilizers commute with gauge and logical operators, and gauge operators commute with logical operators). It can also be seen that the logical operators satisfy the desired algebra of the cyclic quantum double models:
\begin{equation}
        Z_\subfont{L}^aL_\subfont{L}^b = \omega^{ab}L_\subfont{L}^bZ_\subfont{L}^a
\end{equation}
and the gauge operators satisfy an equivalent algebra:
\begin{equation}
        Z_\subfont{G}^aL_\subfont{G}^{-b} = \omega^{ab}L_\subfont{G}^{-b}Z_ \subfont{G}^a \,.
\end{equation}

We define the Hamiltonian on a single code gadget (associated with edge $e$) as
\begin{equation}\label{eqabham1}
        H(e)=-\frac{1}{d}\sum_k\left[L_\subfont{G}^k+(L_\subfont{G}^kS_L^k)^{\dagger}+Z_\subfont{G}^k+(Z_\subfont{G}^kS_Z^k)^{\dagger}\right]_e \,.
\end{equation}
This equation can be represented diagrammatically as
\begin{equation}
H(e)=-\frac{1}{d}\sum_k\left[\siteAb{I}{L^{-k}}{L^{k}}{I}+\siteAb{L^{-k}}{I}{I}{L^{k}}+\siteAb{Z^{-k}}{Z^k}{I}{I}+\siteAb{I}{I}{Z^k}{Z^{-k}}\right] \,.
\end{equation}
In this form, it is easy to see that each term in the Hamiltonian acts only on two qudits.

Multiplication by the stabilizers (as in $L_\subfont{G}S_L$) effectively moves a gauge operator from two qudits onto the opposite two. In the $d=2$ case, this Hamiltonian does not directly reduce to the one quoted for the toric code earlier (Eq.~(\ref{eqTCh0})) because of the inclusion of the identity ($k=0$) term. However this term only induces a constant energy shift, and so can be disregarded for our purposes.

We now turn to the properties of the ground space of this Hamiltonian. It is clear that this Hamiltonian commutes with the logical operators, but less obvious that its ground space possesses the other properties we require. In Appendix~\ref{secProveThmAbGs}, we prove the following theorem: 
\begin{theorem}\label{thmabhamgs}
The Hamiltonian \Eref{eqabham1} has a $d$-fold degenerate ground space that is in the common $+1$ eigenspace of ${S}_L$ and ${S}_Z$.
\end{theorem}
This result, combined with the fact that our logical operators commute with the Hamiltonian, gives us a ground-space to use as an encoded logical codespace $\mathcal{H}_\subfont{L}$.

\subsection{Coupling the Code Gadgets}

The lattice is connected exactly as was the case for the toric code (\Fref{F:toricPEPS}), with qudits from neighbouring edges linked via an entangling bond. We have an unperturbed Hamiltonian for each edge qudit given as in the previous section:
\begin{equation}\label{eqabunpetham}
        H_0=\sum_eH(e)
\end{equation}
where the index $e$ denotes a particular edge qudit.
We then introduce the bond term:
\begin{eqnarray}\label{eqabpert}
        V &=&\sum_b V(b)\\
&=& -\sum_b\sum_{k=0}^{d-1}\left[\bond{L^k}{L^k}+\bond{Z^k}{Z^{-k}}\right]
\end{eqnarray}
coupling the physical qudits connected by bond $b$. The ground state of this bond term is a maximally entangled state of dimension $d$ between the two qudits.

We are interested in reproducing the quantum double Hamiltonian in an encoded form, so to concisely state our objective, we define the encoded $A$ and $B$ operators
\begin{eqnarray}
\hat{A}(v) &\equiv& \frac{1}{|G|} \sum_g\bigotimes_{e \in +(v)} L^g_\subfont{L}(e,v)\\
\hat{B}(p) &\equiv& \sum_{g_k \ldots g_1 = 0}\, \bigotimes_{e_i \in \Box(p)} T_\subfont{L}^{g_i}(e_i,p)
\end{eqnarray}
with $L_\subfont{L}$ defined in \Eref{eqabop1} and $T_\subfont{L}=\frac{1}{d}\sum_{k}\omega^{k}Z^{k}_\subfont{L}$ is the encoded group projection operator. We can then state main result of this section as Theorem \ref{thmabpert}.


\begin{theorem}\label{thmabpert}
The Hamiltonian $H=H_0+\lambda V$ with $H_0$ and $V$ defined as in \Eref{eqabunpetham} and \Eref{eqabpert} on a square lattice has a low energy behaviour described by an effective Hamiltonian of the form
\begin{equation}\label{eqabeffhamthm}
H_{\mathrm{eff}} = c_I I-\left(c_A\lambda^4\right)\sum_v \hat{A}(v) - \left(c_B\lambda^4\right)\sum_p \hat{B}(p) + \mathcal{O}(\lambda^5)
\end{equation}
for some constants $c$ independent of $\lambda$ and $N$, where $N$ is the number of sites on the lattice.
\end{theorem}

The consequence of this theorem is our system's low energy effective Hamiltonian replicating the low-energy sector of the quantum double model (for cyclic groups at this stage). 



\begin{figure}
\centering
\beginpgfgraphicnamed{Figures/encodedAabelian}
\begin{tikzpicture}
        [a/.style={circle, fill=red, scale=.4},
        ao/.style={circle, draw=red, thick, scale=.35},
        b/.style={circle, fill=gray, opacity=.3, scale=.4},
        scale=2.2]

        \newcommand{\dclosed}[1]{\tikz[scale=.5] \fill[#1] (1ex,1ex) circle (1ex);}
        \newcommand{\dopen}[1]{\tikz[scale=.45] \draw[thick,#1] (1ex,1ex) circle (1ex);}

        \node at (.5,.5) [anchor = south east] {$v$};
        \node at (-.6,.5) {$\hat{A}(\dopen{red}, \dclosed{red},v) = \ \ \ $};

        \foreach \a in {0,90,180,270}
                \draw[gray,opacity=.6,rotate around={\a:(.5,.5)}] (0,.5) circle (2mm);

        \def\g{.3}
        \draw[red,ultra thick,shift={(.5,.5)}] (-1+\g,-1+\g) grid (1-\g,1-\g);

        \def\tick{.075}
        \foreach \a in {0,90,180,270}{
                \draw[rotate around={\a:(.5,.5)}] 
                        (.5+\tick,-\tick) node [b] {} (.5-\tick,-\tick) node [b] {};
        }
        
        \foreach \a in {0,90,180,270}{
                \draw[rotate around={\a:(.5,.5)}] 
                        (.5+\tick,\tick) node [a] {} (.5-\tick,\tick) node [ao] {};
        }
\end{tikzpicture}
\endpgfgraphicnamed
\quad \quad
\beginpgfgraphicnamed{Figures/encodedBabelian}
\begin{tikzpicture}
        [a/.style={circle, fill=blue, scale=.4},
        ao/.style={circle, draw=blue, thick, scale=.35},
        yc/.style={circle, fill=red, scale=.4},
        yco/.style={circle, draw=red, thick, scale=.35},
        b/.style={circle, fill=gray, opacity=.4, scale=.4},
        scale=2.2]

        \def\g{.013}
        \draw[blue,ultra thick] (-\g,-\g) grid (1+\g,1+\g);

        \newcommand{\dclosed}[1]{\tikz[scale=.5] \fill[#1] (1ex,1ex) circle (1ex);}
        \newcommand{\dopen}[1]{\tikz[scale=.45] \draw[thick,#1] (1ex,1ex) circle (1ex);}

        \node at (.5,.5) {$p$};
        \node at (-.6,.5) {$\hat{B}(\dopen{blue},\dclosed{blue},p) = \ \ \ $};

        \foreach \a in {0,90,180,270}
                \draw[gray,opacity=.7,rotate around={\a:(.5,.5)}] (0,.5) circle (2mm);

        \def\tick{.075}
        \foreach \a in {0,90,180,270}{
                \draw[rotate around={\a:(.5,.5)}] 
                        (.5+\tick,-\tick) node [b] {} (.5-\tick,-\tick) node [b] {};
        }
        
        \foreach \a in {0,180}{
                \draw[rotate around={\a:(.5,.5)}] 
                        (.5+\tick,\tick) node [a] {} (.5-\tick,\tick) node [ao] {};

        \draw[rotate around={90:(.5,.5)}] 
                (.5+\tick,\tick+0.03) node [a] {} (.5-\tick,\tick+0.03) node [ao] {};
        \draw[rotate around={270:(.5,.5)}] 
                (.5+\tick,\tick-0.005) node [a] {} (.5-\tick,\tick-0.005) node [ao] {};
        \node at (.96,.56) {$^\dagger$};
        \node at (.96,.41) {$^\dagger$};
        \node at (.135,.56) {$^\dagger$};
        \node at (.135,.41) {$^\dagger$};
        }\end{tikzpicture}
\endpgfgraphicnamed
\caption{Physical operators for a cyclic quantum double model. Qudits at locations denoted by open circles will be acted upon by the same single-qudit operator (and similarly for those represented by full circles). The adjoint ($\dagger$) of a given operator is applied to the qudits so labelled. Note the similarity of this diagram with Figure~\ref{F:ToricEncodedOps}, however in contrast to the toric code case, the operators acting on adjacent qubits are now different. \label{F:ToricEncodedOpsAbelian}}
\end{figure}
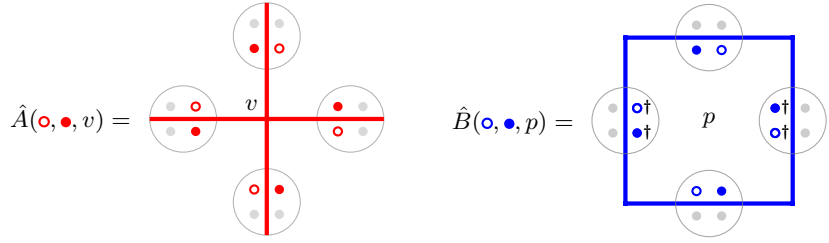

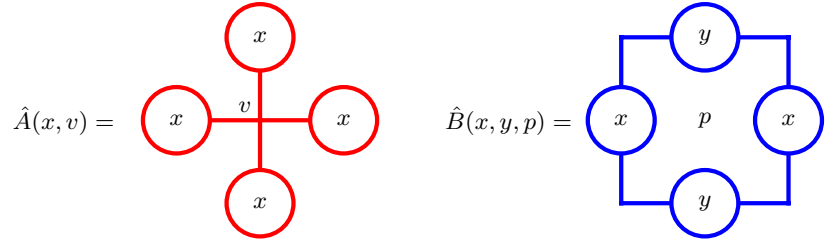
\begin{figure}
\centering
\beginpgfgraphicnamed{Figures/logicalAabelian}
\begin{tikzpicture}[scale=2.2]
        \node at (.5,.5) [anchor = south east] {$v$};
        \node at (-.6,.5) {$\hat{A}(x,v) = \ \ \ $};

        \def\g{.3}
        \draw[red,ultra thick,shift={(.5,.5)}] (-1+\g,-1+\g) grid (1-\g,1-\g);

        \foreach \a in {0,180}
                \filldraw[draw=red, ultra thick, fill=white, rotate around={\a:(.5,.5)}]
                        (0,.5) circle (2mm) node at (0,.5) {$x$};
        \foreach \a in {90,270}
                \filldraw[draw=red, ultra thick, fill=white, rotate around={\a:(.5,.5)}]
                        (0,.5) circle (2mm) node at (0,.5) {$x$};
\end{tikzpicture}
\endpgfgraphicnamed
\quad \quad
\beginpgfgraphicnamed{Figures/logicalBabelian}
\begin{tikzpicture}[scale=2.2]
        \node at (.5,.5) {$p$};
        \node at (-.6,.5) {$\hat{B}(x,y,p) = \ \ \ $};

        \def\g{.013}
        \draw[blue,ultra thick] (-\g,-\g) grid (1+\g,1+\g);

        \foreach \a in {0,90,180,270}
                \filldraw[draw=blue, ultra thick, fill=white, rotate around={\a:(.5,.5)}] 
                        (0,.5) circle (2mm);
                        
        \path node at (0,.5) {$x$} node at (1,.5) {$x$} 
                node at (.5,0) {$y$} node at (.5,1) {$y$};
\end{tikzpicture}
\endpgfgraphicnamed
\caption{Encoded operators for a cyclic quantum double model. Here each of $x$ or $y$ represents a 4-qudit logical or gauge operator. We use an overloaded notation here for $\hat{A}$ and $\hat{B}$ such that when their arguments are encoded 4-qudit operators they take this form, while if the arguments are single qudit operators, then they take the form of \Fref{F:ToricEncodedOpsAbelian}. \label{F:EncodedOpsAbelian}}
\end{figure}

\begin{proof}[Proof of Theorem \ref{thmabpert}]

We again follow the perturbative analysis described in Appendix~\ref{secpertform}. As such, this will require evaluating terms in the perturbative expansion of the self-energy at order $n$:
\begin{equation}
\Sigma^{(n)}(E_0) = \lambda^n\Upsilon V(G_0(E_0)V)^{(n-1)}\Upsilon\\
\end{equation}
with $G_0$ the Green's function for the system (vanishing on ground states) and $\Upsilon$ is the projector to the mutual ground space of each of the code gadgets.

Before we begin our the proof in earnest, it is useful to comment on the kinds of terms which will be preserved and those that will vanish in the ground space. All the operators defined for cyclic groups have some commutation relation $MN = \alpha NM$ for some complex $\alpha$. It is then a simple result to show that for our set of stabilizers any operator which does not commute with each stabilizer will neccessarily excite a ground state to an orthogonal state. This means that an operator $M$ with $\alpha\neq 1$ for $N$ any stabilizer will become a detectable error on the gadget's quantum code, and will take the gadget to an orthogonal subspace. This implies that it will vanish in our perturbative treatment, as terms arising in our effective Hamiltonian are restricted to the ground space. In this way we need only consider error-free terms (i.e.~terms which commute with all stabilizers) in our effective Hamiltonian.

From this discussion, we immediately see that first order terms will vanish, as they will necessarily leave two code gadgets in excited states. The only non-vanishing second order terms will be proportional to identity. In contrast to the toric code as presented earlier, there will be non-vanishing third order terms due to the inclusion of an identity component in the perturbation term \Eref{eqabpert}. However, these terms will be proportional to the second order terms, and so will be trivial. At $4^{\mathrm{th}}$ order we find non-trivial vertex and plaquette terms which survive. To write each of these terms more explicitly, we must now distinguish between inwards directed vertices ($v_+$) and outward directed vertices ($v_-$) as well as the two kinds of plaquettes. We will label the plaquettes left ($p_l$) and right ($p_r$) depending on the orientation of the top edge. In this scheme, the effective Hamiltonian will take the form
\begin{equation}
H_{\mathrm{eff}}^{(4)} = \mathrm{const} - \sum_{k=0}^{d-1}\left[\sum_{v_+}H^k_{v_+}+\sum_{v_-}H^k_{v_-}+ \sum_{p_l}H^k_{p_l}+\sum_{p_r}H^k_{p_r}\right]
\end{equation}

We can write each of these parts individually using the notation of Figure \ref{F:ToricEncodedOpsAbelian}.
\begin{eqnarray}
H^k_{v_+}&=&\kappa_v^L\lambda^4\Upsilon \hat{A}(L^k,L^k,v)\Upsilon+\kappa_v^Z\lambda^4\Upsilon \hat{A}(Z^k,Z^{-k},v)\Upsilon\\
H^k_{v_-}&=&\kappa_v^L\lambda^4\Upsilon \hat{A}(L^{k},L^{k},v)\Upsilon+\kappa_v^Z\lambda^4\Upsilon \hat{A}(Z^{k},Z^{-k},v)\Upsilon\\
H^k_{p_l}&=&\kappa_p^L\lambda^4\Upsilon \hat{B}(L^k,L^{-k},p)\Upsilon+\kappa_p^Z\lambda^4\Upsilon \hat{B}(Z^k,Z^{k},p)\Upsilon\\
H^k_{p_r}&=&\kappa_p^L\lambda^4\Upsilon \hat{B}(L^k,L^{-k},p)\Upsilon+\kappa_p^Z\lambda^4\Upsilon \hat{B}(Z^k,Z^{k},p)\Upsilon
\end{eqnarray}

The $\kappa$ are constants which take into account the sum of products of the Green's functions in the perturbation. They can be calculated for given $d$ once the spectrum of $H_0$ has been found. They must be nonzero for each of these terms because they do not return to the ground state before the end of the perturbation, and so the Green's function will never vanish.

In terms of encoded logical, gauge and stabilizer operators, we can use a similar notation (as seen in \Fref{F:EncodedOpsAbelian}) to write these terms.
\begin{eqnarray}
H^k_{v_+}&=&\kappa_v^L\lambda^4\Upsilon \hat{A}(L^k_\subfont{L}S_L^{-k}, v)\Upsilon+\kappa_v^Z\lambda^4\Upsilon \hat{A}(Z_\subfont{G}^{-k}S_Z^{-k},v)\Upsilon\\
H^k_{v_-}&=&\kappa_v^L\lambda^4\Upsilon \hat{A}(L^{k}_\subfont{L}, v)\Upsilon+\kappa_v^Z\lambda^4\Upsilon \hat{A}(Z^{-k}_\subfont{G}, v)\Upsilon\\
H^k_{p_l}&=&\kappa_p^L\lambda^4\Upsilon \hat{B}(L^k_\subfont{G},L^{-k}_\subfont{G}S_L^{-k},p)\Upsilon+\kappa_p^Z\lambda^4\Upsilon \hat{B}(Z_\subfont{L}^{-k}S_Z^{-k},Z_\subfont{L}^k,p)\Upsilon\\
H^k_{p_r}&=&\kappa_p^L\lambda^4\Upsilon \hat{B}(L_\subfont{G}^{k}S_L^{k},L^{-k}_\subfont{G},p)\Upsilon+\kappa_p^Z\lambda^4\Upsilon \hat{B}(Z_\subfont{L}^{-k},Z_\subfont{L}^{k}S_Z^{k},p)\Upsilon
\end{eqnarray}

We now evaluate the $\Upsilon$ projectors. Stabilizers $S_Z$ \& $S_L\to I$ in the ground space by Theorem \ref{thmabhamgs}, and each gauge operator will simply evaluate as a constant (expectation value) in the ground space. We can again disregard these constant terms as irrelevant energy shifts. This leaves us with:
\begin{eqnarray}
H^k_{v_+} &=& \kappa_v^L\lambda^4\Upsilon \hat{A}(L_\subfont{L}^k,v)\Upsilon\label{eq-Hvab1}\\
H^k_{v_-} &=& \kappa_v^L\lambda^4\Upsilon \hat{A}(L_\subfont{L}^{k},v)\Upsilon\\
H^k_{p_l} &=& \kappa_p^Z\lambda^4\Upsilon \hat{B}(Z_\subfont{L}^{-k},Z_\subfont{L}^{k},p)\Upsilon\\
H^k_{p_r} &=& \kappa_p^Z\lambda^4\Upsilon \hat{B}(Z_\subfont{L}^{-k},Z_\subfont{L}^{k},p)\Upsilon\label{eq-Hpab2}
\end{eqnarray}

By subsitituting the definitions of the $Z$ operators in terms of $T$ projectors, and using the orthogonality of the characters $\omega$, it is simple to verify that these $B$ terms are identical to the $B$ terms of \Eref{E:BpDefinition}. Although the $\pm$ signs associated with the vertices (and edges in general) in these definitions may not immediately seem consistent with the quantum double Hamiltonian presented earlier (\ref{qdtargetham}) for cyclic groups, we have the freedom to rearrange the sums over k in the Hamiltonian. When we consider that we can take $k\to-k$ whenever we like, it becomes clear that these terms are indeed identical to those appearing in the quantum double Hamiltonian. This gives the result that the effective Hamiltonian will take the form \Eref{eqabeffhamthm}.  \qedhere \end{proof}

At all orders $<2L$ ($L$ the smallest linear dimension of the surface into which the model is embedded) the terms in the self-energy expansion (see Appendix A) will act on the logical state like products of the $4^{\mathrm{th}}$ order terms. These terms all commute and will not map the ground space of the quantum double model out of the  $+1$ eigenspace of the encoded vertex and plaquette terms (\ref{eq-Hvab1}-\ref{eq-Hpab2}).

%
%

\section{Our Construction for General Quantum Double Models}\label{secgenqd}

General quantum double models can have a much more complicated algebra than the simple cyclic ones studied previously (see \Sref{secprelimqd}). However, this class includes non-Abelian models which are able to perform universal quantum computation, and so they are of key interest in this study. 

\subsection{Code Gadget Operators}\label{sec-naencoding}

As in the previously studied $\Z_d$ quantum double models, we encode each logical qudit in a code gadget consisting of four strongly coupled physical qudits (the scheme is again as in \Fref{F:toricPEPS}). For a given group $G$, the code gadget on each edge then has a $|G|$-fold degenerate ground space which we use as a logical qudit. The difference in this scheme is that the generalizations of the gauge operators and stabilizers will not commute in general, and so these degrees of freedom are not separable. This is contrary to the normal use of the terms ``stabilizer'' and ``gauge'', but we will abuse the terminology and continue to use these terms in analogy to their cyclic counterparts. The operators we define here directly generalise those used in the previous sections. The differences arise from the non-commutativity of the group multiplication operations and the fact that the irreducible representations of these general groups can be multidimensional (as opposed to the cyclic groups, which have 1-dimensional irreps). With this in mind, we can define logical and ``gauge" operators on the code gadget as follows:
\begin{align}
L^g_{\subfont{L}+}\equiv& \siteN{L^g_+}{L^g_+}{I}{I} \ ,&
T^g_{\subfont{L}+}\equiv& \sum_{g_2g_3 = g}\siteN{I}{T^{g_2}_+}{T^{g_3}_+}{I}\nonumber \ ,\\
L^g_{\subfont{G}-} \equiv& \siteN{I}{L^g_-}{L^g_+}{I} \ ,&
T^g_{\subfont{G}+} \equiv& \sum_{g_1g_2 = g}\siteN{T^{g_1}_-}{T^{g_2}_+}{I}{I} \ .
\end{align}

We can also define operators to describe the ``stabilizer" degrees of freedom:
\begin{align}
S^g_L \equiv& \siteNXL{L^{g^{-1}}_{-}}{L^{g^{-1}}_{-}}{L^{g^{-1}}_{+}}{L^{g^{-1}}_{+}}  \ ,&
S_L \equiv& \frac{1}{|G|}\sum_g S^g_{L} \ , \\
S^g_{T} \equiv& \sum_{g_1g_2g_3g_4 = g} \siteN{T_{-}^{g_1}}{T_{+}^{g_2}}{T_{+}^{g_3}}{T_{-}^{g_4}} \ ,&
S_T \equiv& S^1_{T} \ .
\end{align}
where here the identity group element is denoted $1$. These operators are clearly defined in a very natural way with respect to the quantum double algebra. In these models, the only operators which we strictly require to stabilize (act identically within) the ground space of our edge qudits are the projectors $S_L$ and $S_T$. Despite this, we will extend our notational abuse and continue to refer to $S_L^g$ and $S_T^g$  as ``stabilizer operators''.

The operators we have defined now constitute a minimum operator basis for the degrees of freedom of our encoded qudit. We can also define addititonal operators:
\begin{align}
L^g_{\subfont{L}-}\equiv& \siteN{I}{I}{L^g_-}{L^g_-} \ , &
T^g_{\subfont{L}-} \equiv& \sum_{g_4g_1 = g}\siteN{T^{g_1}_-}{I}{I}{T^{g_4}_-} \ ,  \nonumber\\
L^g_{\subfont{G}+} \equiv& \siteN{L^g_-}{I}{I}{L^g_+} \ ,&
T^g_{\subfont{G}-} \equiv& \sum_{g_3g_4 = g}\siteN{I}{I}{T^{g_3}_+}{T^{g_4}_-} \ .
\end{align}

The $+$ and $-$ subscripts on the joint operators here refer to whether these operators possess the algebra of left or right multiplication (or projection) operators, i.e.
\begin{eqnarray}
L_{\pm}^gT_{\pm}^h &=& T_{\pm}^{gh}L_{\pm}^g\\
L_{\pm}^gT_{\mp}^h &=&T_{\mp}^{hg^{-1}}L_{\pm}^g
\end{eqnarray}

 The stabilizer operators themselves satisfy the quantum double algebra
\begin{equation}
S^g_LS^h_{T} = S^{g^{-1}hg}_{T}S^g_L
\end{equation}

The logical operators commute with both the gauge and stabilizer operators, but in this scheme the gauge and stabilizer operators do not commute with each other.

We can also derive the $Z$ operators corresponding to our $T$ operators (see \Sref{secprelimqd}):
\begin{align}
Z^{\pi_{ij}}_{\subfont{L}+} =& \sum_{m}\siteNXL{I}{Z^{\pi_{im}}_+}{Z^{\pi_{mj}}_+}{I} \ ,&
Z^{\pi_{ij}}_{\subfont{G}+} =& \sum_{m}\siteNXL{Z^{\pi_{im}}_-}{Z^{\pi_{mj}}_+}{I}{I} \nonumber \\
Z^{\pi_{ij}}_{\subfont{L}-}=& \sum_m  \siteNXL{Z_-^{\pi_{mj}}}{I}{I}{Z_-^{\pi_{im}}}  ,&
Z^{\pi_{ij}}_{\subfont{G}-} =& \sum_m \siteNXL{I}{I}{Z_+^{\pi_{im}}}{Z_-^{\pi_{mj}}}\ .\label{eq-zopnadefine}
\end{align}
By construction, these operators satisfy the relevant algebra for $Z$ operators.


We have now defined operators on the code gadget corresponding to both the left regular representation and the right regular representation of the group $G$ for the encoded gauge and logical qudits. These definitions are redundant, in that the gauge or logical state can be uniquely defined through the action of only one $L$ operator and only one $T$ operator (or equivalently $Z$ operator). This gives us some freedom in the definition of the logical state. We will choose to define it through the action of $L^{g}_{\subfont{L}-}$ and $T^{g}_{\subfont{L}+}$. The action of $L^{g}_{\subfont{L}+}$ and $T^{g}_{\subfont{L}-}$ will then be poorly defined in general with respect to the logical state, but we will choose the codespace such that they act appropriately within it.

In order to understand how these operators act on the logical state, it is useful at this point to explicitly identify the encoding scheme of our code gadget. We can label the state of each physical qudit within a given code gadget as follows:
\begin{equation}\label{eq-physicalbasis}
\left|h_a,h_b,k_{\subfont{G}},k_{\subfont{L}}\right> = \siteN{\left|h_a\right>}{\left|k_{\subfont{G}}h_a\right>}{\left|h_a^{-1}k^{-1}_{\subfont{G}}k_{\subfont{L}}\right>}{\left|h_b^{-1}h_a^{-1}k_{\subfont{L}}\right>}
\end{equation}

Here the logical and gauge states of the system are labelled by $k_\subfont{L}$ and $k_\subfont{G}$ respectively, and will transform under the the action of the logical or gauge operators. The remaining labels $h$ define how the stabilizers $S_{\subfont{L}}$ and $S_{\subfont{G}}$ act on the system.
We can directly see the action of our encoded operators on these states. For example,
\begin{eqnarray}
L^{g}_{\subfont{L}-}\ket{h_a,h_b,k_{\subfont{G}},k_{\subfont{L}}}&=&\ket{h_a,h_b,k_{\subfont{G}},k_{\subfont{L}}g^{-1}}\\
T^{g}_{\subfont{L}+}\ket{h_a,h_b,k_{\subfont{G}},k_{\subfont{L}}}&=& \delta_{gk_{\subfont{L}}}\ket{h_a,h_b,k_{\subfont{G}},k_{\subfont{L}}} .
\end{eqnarray}

We chose to define the logical state $k_\subfont{L}$ by the action of these two operators, and as such they act on $k_\subfont{L}$ as might be expected. The stabilizer operators act as
\begin{eqnarray}
S^g_L\left|h_a,h_b,k_{\subfont{G}},k_{\subfont{L}}\right>&=&\left|h_ag,g^{-1}h_bg,k_{\subfont{G}},k_{\subfont{L}}\right>\\
S^g_T\left|h_a,h_b,k_{\subfont{G}},k_{\subfont{L}}\right> &=& \delta_{gh_b}\left|h_a,h_b,k_{\subfont{G}},k_{\subfont{L}}\right> \, .
\end{eqnarray}

The action of the remaining encoded operators is not so simple. In general, these will mix the logical or gauge states with the $h_a$ or $h_b$. However, we will construct the codespace such that within it, logical operators will act only on the logical state, and gauge operators similarly act appropriately.

\subsection{Code Gadget Hamiltonian}\label{subsecqdham}

Now that we have defined a set of operators we consider the Hamiltonian of a code gadget. We require that this Hamiltonian consist only of 2-body terms and possess a ground space which can be used as a codespace for our logical qudit. For this to happen, the ground space must be stabilized by $S_L$ and $S_T$, and be exactly $|G|$-fold degenerate. The significance of these two stabilizers is that they are terms in the quantum double Hamiltonian defined on a small 4-qudit torus (as depicted in \Fref{F:SmallTorus}). In that sense, we are creating a code which is very similar to a miniature quantum double model. Of course, our code gadget will consist of only two-body terms, and to achieve this we must sacrifice the gauge degrees of freedom in our model. 

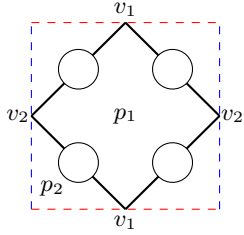
\begin{figure}
\beginpgfgraphicnamed{Figures/smalltorus}
\begin{tikzpicture}[scale=1.75]
        \draw[dashed,color=blue] (-0.20716,-0.20716) -- (-0.20716,1.20716);
        \draw[dashed,color=red] (-0.20716,-0.20716) -- (1.20716,-0.20716);
        \draw[dashed,color=red] (-0.20716,1.20716) -- (1.20716,1.20716);
        \draw[dashed,color=blue] (1.20716,-0.20716) -- (1.20716,1.20716);
        \draw[step=1cm,rotate around={45:(.5,.5)}, thick] (0,0) grid (1,1);
        \foreach \x in {0} \foreach \y in {0,1}
                \filldraw[rotate around={45:(.5,.5)},xshift=\x cm+.5cm,yshift=\y cm,color=white,draw=black]
                (0,0) circle (1.5mm);
        \foreach \x in {0,1} \foreach \y in {0}
                \filldraw[rotate around={45:(.5,.5)},xshift=\x cm,yshift=\y cm+.5cm,color=white,draw=black]
                (0,0) circle (1.5mm);
        \draw (.5,.5) node{$p_1$};
        \draw (-0.05,-0.05) node{$p_2$};
        \draw (.5,1.30716) node{$v_1$};
        \draw (.5,-0.30716) node{$v_1$};
        \draw (1.30716,.5) node{$v_2$};
        \draw (-0.30716,.5) node{$v_2$};
\end{tikzpicture}
\endpgfgraphicnamed
\caption{The 4-qudit torus is constructed by identifying the opposite vertices of a square. The dashed horizontal (red) lines are identified, and the dashed vertical (blue) lines are identified. There are two equivalent plaquette stabilizers corresponding to $p_1$ and $p_2$ and two (generally inequivalent) vertex stabilizers corresponding to $v_1$ and $v_2$ in this quantum double model. Our stabilizers $S_L$ and $S_T$ are obtained by choosing one of each. For Abelian models the two vertex stabilizers are also equivalent, and so $S_L$ and $S_T$ generate the full stabilizer group of the quantum double model on this surface for those groups.}\label{F:SmallTorus}
\end{figure}


To this end, we generalize the Hamiltonian we used for cyclic groups previously to give
\begin{eqnarray}\label{eqnahamunpet}
H(e) &=& -\frac{J_L}{|G|}\sum_g\left[L^g_{\subfont{G}+}+L^{g}_{\subfont{G}-}\right] 
-\frac{J_Z}{|G|} \sum_{\pi,i}d_\pi\left[Z^{\pi_{ii}}_{\subfont{G}+}+Z^{\pi_{ii}}_{\subfont{G}-}\right]\label{eqnaham}\\
&=& -\frac{J_L}{|G|}\sum_g\left[L^g_{\subfont{G}+}+L^{g}_{\subfont{G}-}\right] 
-J_Z\left[T^{1}_{\subfont{G}+}+T^{1}_{\subfont{G}-}\right]
\end{eqnarray}

These two forms can be seen to be equivalent using the definitions of the $Z_\subfont{G}$ operators (\Eref{eq-zopnadefine}). In the latter form, it is clear that our unperturbed Hamiltonian is very closely related to the quantum double Hamiltonian (\Eref{qdtargetham}). It is also clear that the logical operators will commute with the Hamiltonian (as each term is a gauge operator), and so we can expect at least the $|G|$-fold degeneracy of our logical qubit. The following theorem encapsulates the remaining requirements of our codespace.
\begin{theorem}\label{thmnags}
The Hamiltonian \Eref{eqnahamunpet} has a $|G|$-fold degenerate ground space that is in the common $+1$ eigenspace of $S_L$ and $S_T$.
\end{theorem}
We provide a proof of this theorem in \Aref{provethmnags}. This theorem places restrictions on the form of the ground space which ensure that undesirable terms in the effective Hamiltonian will couple to higher energy sectors and vanish.  Now we have a suitable codespace for our logical qudit, we proceed to reproducing the target model.

\subsection{Coupling the Code Gadgets}\label{subsecqdcouple}

We place our code gadgets on the edges of the lattice and couple them using the same geometric scheme as the toric code and cyclic group cases (illustrated in \Fref{F:toricPEPS}), where each physical qudit is perturbatively coupled to a physical qudit from a neighbouring code gadget. Our uncoupled (unperturbed) Hamiltonian for each edge of the lattice is as given in \Eref{eqnahamunpet}.
\begin{eqnarray}\label{eqnaunpetham}
H_0 &=& \sum_eH(e)\\
&=&\sum_e\left[-\frac{J_L}{|G|}\sum_g\left[L^g_{\subfont{G}+}+L^{g}_{\subfont{G}-}\right] 
-\frac{J_Z}{|G|} \sum_{\pi,i}d_\pi\left[Z^{\pi_{ii}}_{\subfont{G}+}+Z^{\pi_{ii}}_{\subfont{G}-}\right]\right]_e
\end{eqnarray}

As before, the subscript $e$ refers to a particular edge (a particular code gadget). The perturbation term is generalized straightforwardly from the cyclic case to take the form:
\begin{eqnarray}\label{eqnapert}
V&=&\sum_bV(b)\\
&=& -\sum_b\left[\sum_k\bond{L^k}{L^k}+\sum_{\pi,k,m} d_\pi \bond{Z^{\pi_{km}}}{Z^{\pi_{mk}}}\right]
\end{eqnarray}
where the $L$ or $Z$ operator associated with a particular qudit has definite $\pm$ subscript depending on its location as follows
\begin{align}
L:&\siteN{+}{+}{-}{-} &
Z:&\siteN{-}{+}{+}{-} \nonumber
\end{align}
e.g.~an $L$ operator in the top right hand corner of a code gadget (with the edge taken running up the page) will take the form $L_+$. This is motivated by the definitions of $v_\pm$ and $p_\pm$ of the quantum double model in \Sref{secprelimqd}. As in the $\Z_d$ models, it is clear that the coupling terms are very closely related to the quantum double Hamiltonian. In fact, these bond terms could be considered a quantum double Hamiltonian on a small (2-qudit) sphere, with the corresponding non-degenerate ground state (see \Fref{F:SmallSphere}).

\begin{figure}
\beginpgfgraphicnamed{Figures/smallsphere}
\begin{tikzpicture}[scale=1.75]
        \draw[dashed,color=red] (-0.20716,-0.20716) -- (-0.20716,1.20716);
        \draw[dashed,color=red] (-0.20716,-0.20716) -- (1.20716,-0.20716);
        \draw[dashed,color=red] (-0.20716,1.20716) -- (1.20716,1.20716);
        \draw[dashed,color=red] (1.20716,-0.20716) -- (1.20716,1.20716);

        \draw (.5,.5) node{$p_1$};
        \draw (0.02,0.00) node{$p_2$};
\draw[thick] (0.5,1.02) arc (60:-60:0.6cm);
\draw[thick] (0.5,1.02) arc (120:240:0.6cm);
        \foreach \x in {0.2,0.8}
                \filldraw[xshift=\x cm,color=white,draw=black]
                (0,0.5) circle (1.5mm);
        \draw (.5,1.1) node{$v_1$};
        \draw (.5,-0.13) node{$v_2$};
\end{tikzpicture}
\endpgfgraphicnamed
\caption{The 2-qudit sphere is constructed by identifying every point on the boundary of a plane. This boundary is represented by dashed red lines. The solid lines correspond to edges on the lattice, with a qudit located on each edge. The loop formed by these edges is homologically equivalent to an equator around the sphere. There exists a plaquette stabilizer for each plaquette $p_1$ and $p_2$ and similarly a vertex stabilizer for each vertex $v_1$ and $v_2$ in this quantum double model. The terms in our bond Hamiltonian are obtained by choosing one plaquette and one vertex stabilizer from these. For Abelian models the two vertex stabilizers are also equivalent, as are the two plaquette stabilizers, and so the bond Hamiltonian consists of a full generating set of the stabilizer group of the quantum double model on this surface for these groups.}\label{F:SmallSphere}
\end{figure}
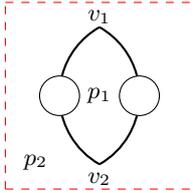

Again we must define the encoded $A$ and $B$ operators we will reproduce in our perturbative expansion (with $1$ the identity group element for general groups)
\begin{eqnarray}
\hat{A}(v) &\equiv& \frac{1}{|G|} \sum_g\bigotimes_{e \in +(v)} L^g_\subfont{L}(e,v)\\
\hat{B}(p) &\equiv& \sum_{g_k \ldots g_1 = 1}\, \bigotimes_{e_i \in \Box(p)} T_\subfont{L}^{g_i}(e_i,p)
\end{eqnarray}

The main result of this section is Theorem \ref{thmnapert}.


\begin{theorem}\label{thmnapert}
The Hamiltonian $H=H_0+\lambda V$ with $H_0$ and $V$ defined as in \Eref{eqnaunpetham} and \Eref{eqnapert} on a square lattice has a low energy behaviour described by an effective Hamiltonian of the form
\begin{equation}\label{eqnaeffhamthm}
H_{\mathrm{eff}} = c_I I-\left(c_A\lambda^4\right)\sum_v \hat{A}(v) - \left(c_B\lambda^4\right)\sum_p \hat{B}(p) + \mathcal{O}(\lambda^5)
\end{equation}
for constants $c$ independent of $\lambda$ and $N$, where $N$ is the number of sites on the lattice. The encoded low energy behavior of the system to this order is described by the quantum double Hamiltonian \Eref{qdtargetham} up to additive and multiplicative constants.
\end{theorem}

This theorem implies that our effective Hamiltonian correctly reproduces the low energy sector of the quantum double Hamiltonian for any group. Before we prove it, we make some comment on the operators that will arise in the perturbative treatment. For cyclic models, we were able to make some general arguments to exclude all unwanted operators from arising in the effective Hamiltonian. Unfortunately, no equivalent general argument has been found for the general quantum double models. For this reason, we have explicitly shown this property for the relevant operators in \Aref{secerrorops}. The results are an intuitive generalization of the cyclic case, in that the only non-vanishing operators will be those contributing encoded  operators which act on the gauge or logical subspaces. We use these results now to calculate the non-vanishing terms in the effective Hamiltonian.

\begin{proof}[Proof of Theorem \ref{thmnapert}]

We follow the perturbative treatment as used previously (see \Aref{secpertform}). The results of \Aref{secerrorops} effectively show that no $2^{\mathrm{nd}}$ order terms are able to survive except those proportional to identity (exactly as in the cyclic case). Similarly, all first order terms, and all non-trivial third order terms will vanish. It is clear then, that no non-trivial operators will appear below 4th order. The effective Hamiltonian at order $4$ is then given by
\begin{equation}
H_{\mathrm{eff}}^{(4)} = \lambda^n\Upsilon V(G_0(E_0)V)^{(n-1)}\Upsilon \, ,\\
\end{equation}
$G_0$ here the Green's function vanishing on ground states (taken at the unperturbed ground state energy $E_0$) and $\Upsilon$ again the projector to the ground space of the code gadgets.
At this order we will find non-trivial terms around plaquettes and vertices consisting of products of $Z$ or $L$ operators. As in the $\Z_d$ case, it is useful to distinguish the two different kinds of vertices and the two different kinds of plaquettes here (see \Fref{figdirlattice}). We have defined the lattice such that vertices can either consist of all inwardly directed edges ($v_+$) or all outwardly directed edges ($v_-$). Plaquettes can either have their top edge directed left ($p_l$) or right ($p_r$).

Disregarding constant energy shifts, we can then separate terms in the Hamiltonian by vertex or plaquette type:
\begin{eqnarray}
H_{\mathrm{eff}}^{(4)} &=& -\Upsilon\left[\sum_{v_+}H_{v_+}+\sum_{v_-}H_{v_-}+\sum_{p_l}H_{p_l}+\sum_{p_r}H_{p_r}\right]\Upsilon
\end{eqnarray}

As we did for the cyclic groups, we will write these operators pictorially to make the physical location of a particular operator obvious. For simplicity, we have neglected to draw those qudits on the outer edge of the vertex/plaquette under consideration, where these terms will act trivially. With this in mind, the individual terms in the effective Hamiltonian can be written as:
\begin{equation}
\label{eqnaopfirst}
H_{v_+}=\kappa_{v}^L\lambda^4
 \sum_g\vertexNA{L^g_+}{L^g_+}{L^g_+}{L^g_+}{L^g_+}{L^g_+}{L^g_+}{L^g_+}{1}
+\kappa_v^Zd_\pi^4\lambda^4\sum_{\pi,i,r,m,n}\sum_{j}
\vertexNA{Z_+^{\pi_{rj}}}{Z_-^{\pi_{ji}}}{Z_+^{\pi_{ij}}}{Z_-^{\pi_{jm}}}{Z_+^{\pi_{mj}}}{Z_-^{\pi_{jn}}}{Z_+^{\pi_{nj}}}{Z_-^{\pi_{jr}}}{1}
\end{equation}
\begin{equation}
H_{v_-}=\kappa_v^L\lambda^4
 \sum_g\vertexNA{L^g_-}{L^g_-}{L^g_-}{L^g_-}{L^g_-}{L^g_-}{L^g_-}{L^g_-}{-1}
+\kappa_v^Zd_\pi^4\lambda^4\sum_{\pi,i,r,m,n}\sum_{j}
\vertexNA{Z_-^{\pi_{rj}}}{Z_+^{\pi_{ji}}}{Z_-^{\pi_{ij}}}{Z_+^{\pi_{jm}}}{Z_-^{\pi_{mj}}}{Z_+^{\pi_{jn}}}{Z_-^{\pi_{nj}}}{Z_+^{\pi_{jr}}}{-1}
\end{equation}
\begin{equation}
H_{p_l}=\kappa_p^L\lambda^4
 \sum_g\plaquetteNA{L^g_+}{L^{k_\subfont{L}^{-1} gk_\subfont{L}}_-}{L^{k_\subfont{L}^{-1} gk_\subfont{L}}_-}{L^g_+}{L^g_+}{L^{k_\subfont{L}^{-1} gk_\subfont{L}}_-}{L^{k_\subfont{L}^{-1} gk_\subfont{L}}_-}{L^g_+}{1}
+\kappa_p^Zd_\pi^4\lambda^4\sum_{\pi,i,r,m,n}\sum_{j}
\plaquetteNA{Z_-^{\pi_{jr}}}{Z_-^{\pi_{ij}}}{Z_+^{\pi_{ji}}}{Z_+^{\pi_{mj}}}{Z_-^{\pi_{jm}}}{Z_-^{\pi_{nj}}}{Z_+^{\pi_{jn}}}{Z_+^{\pi_{rj}}}{1}
\end{equation}
\begin{equation}
H_{p_r}=\kappa_p^L\lambda^4
 \sum_g\plaquetteNA{L^{k_\subfont{L}^{-1} gk_\subfont{L}}_-}{L^g_+}{L^g_+}{L^{k_\subfont{L}^{-1} gk_\subfont{L}}_-}{L^{k_\subfont{L}^{-1} gk_\subfont{L}}_-}{L^g_+}{L^g_+}{L^{k_\subfont{L}^{-1} gk_\subfont{L}}_-}{-1}
+\kappa_p^Zd_\pi^4\lambda^4\sum_{\pi,i,r,m,n}\sum_{j}
\plaquetteNA{Z_+^{\pi_{jr}}}{Z_+^{\pi_{ij}}}{Z_-^{\pi_{ji}}}{Z_-^{\pi_{mj}}}{Z_+^{\pi_{jm}}}{Z_+^{\pi_{nj}}}{Z_-^{\pi_{jn}}}{Z_-^{\pi_{rj}}}{-1}\label{eqnaoplast}
\end{equation}

Here the $\kappa$'s are proportionality constants arising from the pertubation treatment. For a specific model, they can readily be calculated. It should be reasonably clear that no other non-trivial $4^{\mathrm{th}}$ order terms (or lower) than those shown above will survive (see \Aref{secerrorops}). The higher order terms will act in the logical space as products of the $4^{\mathrm{th}}$ order terms, until the order is sufficiently high to form non-contractible loops over the lattice.

The most illustrative of the $4^{\mathrm{th}}$ order terms are the plaquette $Z$ terms. As such, we will present a brief demonstration of how these operators arise. On our plaquette, at 4th order the most general $Z$ operator to be constructed has the form:
\[\Upsilon \;
\plaquetteNA{Z_-^{\pi_{sr}}}{Z_-^{\pi_{ij}}}{Z_+^{\pi_{ji}}}{Z_+^{\pi_{mn}}}{Z_-^{\pi_{nm}}}{Z_-^{\pi_{tq}}}{Z_+^{\pi_{qt}}}{Z_+^{\pi_{rs}}}{1}
\;\Upsilon
\]


Here we have already used the fact that $Z$ operators with different representations will move out of the ground space, and so we have only used one representation $\pi$. From Eq. (\ref{zerroreq1}, \ref{zerroreq2}) we find the condition $j=s=n=q$ for non-vanishing terms. From this, we immediately obtain the term appearing above in the effective Hamiltonian for this plaquette. Similar considerations give the other plaquette and vertex terms.

Given the operators (\ref{eqnaopfirst}-\ref{eqnaoplast}) in their current form, it is difficult to immediately see how many of them act on the logical state. In order to study this, we must revisit the encoding scheme defined in \Eref{eq-physicalbasis} and the ground space studied in \Aref{provethmnags}. We will study the action of each operator on the basis 
\begin{equation}
\ket{\sigma_{mn},h,k_\subfont{G},k_\subfont{L}} \equiv \frac{\sqrt{d_\sigma }}{\sqrt{|G|}}\sum_{g_\sigma  \in \mathcal{G}}[\sigma(g_\sigma )]_{mn}\ket{g_\sigma ,g_\sigma ^{-1}hg_\sigma ,k_\subfont{G},k_\subfont{L}}
\end{equation}
particularly in the ground space where $\sigma_{mn} = I_{11}$ and $h=1$. From this we can examine the effect on the logical state $k_\subfont{L}$ in the codespace. The operators in question act on these states as follows
\begin{eqnarray}\label{eqopstudyfirst}
\siteGIGANTOR{L^g_+}{I}{I}{L^{k_\subfont{L}^{-1} gk_\subfont{L}}_-}\ket{I_{11},1,k_\subfont{G},k_\subfont{L}} &=&
 \frac{1}{\sqrt{|G|}}\sum_{g_\sigma}\ket{gg_\sigma,1,k_\subfont{G}g^{-1},k_\subfont{L}}\\
&=& \frac{1}{\sqrt{|G|}}\sum_{\tilde{g}_\sigma}\ket{\tilde{g}_\sigma,1,k_\subfont{G}g^{-1},k_\subfont{L}}\\
&=&\ket{I_{11},1,k_\subfont{G}g^{-1},k_\subfont{L}}
\end{eqnarray} 
and
\begin{eqnarray}
\siteGIGANTOR{I}{L^g_+}{L^{k_\subfont{L}^{-1} gk_\subfont{L}}_-}{I}\ket{I_{11},1,k_\subfont{G},k_\subfont{L}} &=&
 \frac{1}{\sqrt{|G|}}\sum_{g_\sigma}\ket{g_\sigma,1,gk_\subfont{G},k_\subfont{L}}\\
&=&\ket{I_{11},1,gk_\subfont{G},k_\subfont{L}}
\end{eqnarray} 

In the codespace, these operators have the effective action equivalent to $L^g_\subfont{G-}$ or $L^g_\subfont{G+}$ respectively. They differ from the previously defined versions, but the key is that they act only on the gauge state in this subspace. We will denote these operators $L^g_\subfont{G'-}$ and $L^g_\subfont{G'+}$. In a similar way, we can look at the action of the $Z$-like operators:
\begin{eqnarray}
\sum_m\siteNXL{Z_-^{\pi_{mj}}}{Z_+^{\pi_{im}}}{I}{I}\ket{I_{11},1,k_\subfont{G},k_\subfont{L}} 
&=&[\pi(k_\subfont{G})]_{ij}\ket{I_{11},1,k_\subfont{G},k_\subfont{L}} 
\end{eqnarray}

This operator acts quite similarly (identical up to conjugacy) to $Z_{\subfont{G}+}^{\pi_{ij}}$. In the ground space, we will be able to ignore it as a constant expectation value (as the gauge state will be fixed). As such, we will denote this operator $Z_{\subfont{G}'+}^{\pi_{ij}}$. The second of the $Z$ operators is a little more complicated. To this end, we must consider the exact form of the ground states:
\begin{equation}
\ket{\psi_0} = \ket{I_{11},1,1,k_\subfont{L}} + \ket{I_{11},1,I_{11},k_\subfont{L}}
\end{equation}

This allows us to see the action of these operators.
\begin{eqnarray}
\sum_m\siteNXL{I}{I}{Z_+^{\pi_{mj}}}{Z_-^{\pi_{im}}}\ket{I_{11},1,1,k_\subfont{L}}
&=&[\pi(k_\subfont{L}^{-1}k_\subfont{L})]_{ij}\ket{I_{11},1,1,k_\subfont{L}}\\
&=&\delta_{ij}\ket{I_{11},1,1,k_\subfont{L}}\\
\label{eq-mixedgauge}\sum_m\siteNXL{I}{I}{Z_+^{\pi_{mj}}}{Z_-^{\pi_{im}}}\ket{I_{11},1,I_{11},k_\subfont{L}}
&=&\frac{1}{\sqrt{|G|}}\sum_{g_\pi}[\pi(k_\subfont{L}^{-1}g_\pi^{-1}k_\subfont{L})]_{ij}\ket{I_{11},1,g_\pi,k_\subfont{L}}\label{eqopstudylast}
\end{eqnarray}

On the ground space, this acts not dissimilarly to $Z^{\pi_{ij}}_\subfont{G-}$. The gauge state becomes mixed up somewhat, but because of the particular definite gauge state we have in the code space, an operator of this form can be evaluated as a constant. This is a result of the fact that the overlap of this state (\Eref{eq-mixedgauge}) with the ground space is independent of $k_\subfont{L}$ (as can be easily verified). With this in mind, we will write this operator as $Z^{\pi_{ij}}_\subfont{G'-}$ from now on.

Given these definitions, we can now write the terms in the effective Hamiltonian in a more succinct format:
\begin{equation}
\label{opform1}
H_{v_+}=\kappa_v^L\lambda^4\sum_{g}\vertexNAenc{L^g_{\subfont{L}+}}{L^g_{\subfont{L}+}}{L^g_{\subfont{L}+}}{L^g_{\subfont{L}+}}{1}
+\kappa_v^Z\lambda^4\sum_{\pi,i,r,m,n}d_\pi\vertexNAenc{Z_{\subfont{G}'+}^{\pi_{ri}}}{Z_{\subfont{G}'+}^{\pi_{im}}}{Z_{\subfont{G}'+}^{\pi_{mn}}}{Z_{\subfont{G}'+}^{\pi_{nr}}}{1}
\end{equation}
\begin{equation}
H_{v_-}=\kappa_v^L\lambda^4\sum_{g}\vertexNAenc{L_{\subfont{L}-}^g}{L_{\subfont{L}-}^g}{L_{\subfont{L}-}^g}{L_{\subfont{L}-}^g}{-1}
+\kappa_v^Z\lambda^4\sum_{\pi,i,r,m,n}d_\pi\vertexNAenc{Z_{\subfont{G}'-}^{\pi_{ri}}}{Z_{\subfont{G}'-}^{\pi_{im}}}{Z_{\subfont{G}'-}^{\pi_{mn}}}{Z_{\subfont{G}'-}^{\pi_{nr}}}{-1}
\end{equation}
\begin{equation}
H_{p_l}=\kappa_p^L\lambda^4\sum_{g}\plaquetteNAenc{L_{\subfont{G'}-}^g}{L_{\subfont{G'}+}^g}{L_{\subfont{G'}-}^g}{L_{\subfont{G'}+}^g}{1}
+\kappa_p^Z\lambda^4\sum_{\pi,i,r,m,n}d_\pi\plaquetteNAenc{Z_{\subfont{L}-}^{\pi_{ir}}}{Z_{\subfont{L}+}^{\pi_{rn}}}{Z_{\subfont{L}-}^{\pi_{nm}}}{Z_{\subfont{L}+}^{\pi_{mi}}}{1}
\end{equation}
\begin{equation}
H_{p_r}=\kappa_p^L\lambda^4\sum_{g}\plaquetteNAenc{L_{\subfont{G'}+}^g}{L_{\subfont{G'}-}^g}{L_{\subfont{G'}+}^g}{L_{\subfont{G'}-}^g}{-1}
+\kappa_p^Z\lambda^4\sum_{\pi,i,r,m,n}d_\pi\plaquetteNAenc{Z_{\subfont{L}+}^{\pi_{ir}}}{Z_{\subfont{L}-}^{\pi_{rn}}}{Z_{\subfont{L}+}^{\pi_{nm}}}{Z_{\subfont{L}-}^{\pi_{mi}}}{-1}\label{opform8}
\end{equation}

We have been able to treat each code gadget in the $Z$ terms seperately, even though in the previous definitions they shared a common index ($j$). The reason that we are able to separate them in this way is because in the ground space for any given $j$ we can use Eq. (\ref{zerroreq1} - \ref{zerroreq4}) to rewrite each term as an average over all values of $j$.  This allows us to rewrite each term as a summation over different indices, which gives us their quoted form. The factor of $d_\pi$ appearing in the terms (\ref{opform1}-\ref{opform8}) then comes from the remaining sum over $j$. As a simplified explicit example, consider:
\begin{equation}
\Upsilon \sum_j\siteNXL{Z_-^{\pi_{jm}}}{Z_+^{\pi_{ij}}}{I}{I}\otimes
\siteNXL{Z_-^{\pi_{jp}}}{Z_+^{\pi_{mj}}}{I}{I}\Upsilon = \Upsilon \sum_j \sum_{j'}\frac{1}{d_{\pi}}\siteNXL{Z_-^{\pi_{j'm}}}{Z_+^{\pi_{ij'}}}{I}{I}\otimes
\siteNXL{Z_-^{\pi_{jp}}}{Z_+^{\pi_{mj}}}{I}{I}\Upsilon
\end{equation}

We can disregard all the terms in the Hamiltonian which do not act on the logical subspace (as they will only introduce some constant energy shift for our purposes). We now also drop the $\pm$ subscript on logical operators, with the observation that they are consistent with the edge orientation conventions introduced in \Sref{secprelimqd}. With this in mind, our effective Hamiltonian reduces to:
\begin{equation}\label{eq-finalnaeffham}
H_{\mathrm{eff}} = -\lambda^4\Upsilon \left[\sum_{v}\kappa_v^L\sum_g\vertexNAencNoArrow{L^g_{\subfont{L}}}{L^g_{\subfont{L}}}{L^g_{\subfont{L}}}{L^g_{\subfont{L}}}+\sum_{p}\kappa_p^Z\sum_{\pi,i,r,m,n}d_\pi\plaquetteNAencNoArrow{Z_{\subfont{L}}^{\pi_{ir}}}{Z_{\subfont{L}}^{\pi_{rn}}}{Z_{\subfont{L}}^{\pi_{nm}}}{Z_{\subfont{L}}^{\pi_{mi}}}\right]\Upsilon
\end{equation}
with each operator defined as in Eq. (\ref{opform1} - \ref{opform8}), and the $\pm$ of the logical operators defined by the edge orientation. Using the definitions of the $Z$ operators in terms of $T$ operators, and the orthogonality of group characters, it is not difficult to show that these operators are indeed equivalent to those of the target quantum double Hamiltonian. This gives our Hamiltonian the form \Eref{eqnaeffhamthm} as claimed.\qedhere
\end{proof}

As in the Abelian case,  all orders $<2L$ ($L$ the smallest linear dimension of the surface into which the model is embedded) of the self-energy expansion (see Appendix A) will act on the logical state like products of the $4^{\mathrm{th}}$ order terms. These terms all commute and will not map the ground space of the quantum double model out of the  $+1$ eigenspace of the encoded vertex and plaquette terms (\ref{eq-Hvab1}-\ref{eq-Hpab2}). Beyond this order perturbative corrections to the self-energy will be able to form homologically non-trivial loops on the surface. These will be heavily suppressed for large $L$ and large energy gap. 
\appendix
%
%

\section{Perturbation Theory}\label{secpertform}
We will now give a brief introduction to the formalism we will use to perform perturbation calculations, such as to introduce lattice couplings between edge qudits. We follow the resolvent or Green's function approach in~\cite{Kitaev-06} and in general we are only interested in the leading non-constant order in the effective Hamiltonian. Given the Hamiltonian
\begin{equation}
H=H_0+\lambda V
\end{equation}
where $H_0$ has a subspace of degenerate eigenvectors with energy $E_0$. Let $\Upsilon$ be the projector onto the eigenspace of the eigenvalue $E_0$ of $H_0$. In our case we are interested in the situation where $E_0$ is the ground state energy of $H_0$. In degenerate perturbation theory one generally aims to find an effective Hamiltonian $H_{\mathrm{eff}}$ that acts on the subspace given by $\Upsilon$ and that has the same eigenvalues as $H$, in other words, an effective Hamiltonian that describes how the perturbation term $V$ acts within the ground space of the unperturbed Hamiltonian. We use the Green's function formalism of~\cite{Kitaev-06} for this calculation and find the self-energy $\Sigma(E)$ which is given by the perturbation expansion:
\begin{equation}
\Sigma(E)=\Upsilon\left(\lambda V+\lambda^2VG_0(E)V+\lambda^3VG_0(E)VG_0(E)V+\cdots\right)\Upsilon
\end{equation}
At the lowest non-trivial order of perturbation theory  we have
\begin{equation}
H_{\mathrm{eff}} \simeq E_0+\Sigma(E_0)
\end{equation}
The unperturbed Green's function for excited states of $H_0$ is denoted by $G_0(E)=(E-H_0)^{-1}(1-\Upsilon)$, such that this function vanishes in the ground space.
At higher orders in perturbation theory, one would need to take into account the $E$ dependence of the self-energy around $E\approx E_0$ in order to find the effective Hamiltonian; however, as we are interested only in the lowest non-trivial order in perturbation theory throughout this paper, we do not.  

We can expand the self-energy order by order as follows (neglecting constant energy shifts):
\begin{eqnarray}
\Sigma(E_0) &=& \sum_n \Sigma^{(n)}(E_0)\\
\Sigma^{(n)}(E_0) &=& \lambda^n\Upsilon V(G_0(E_0)V)^{(n-1)}\Upsilon
\end{eqnarray}
The Green's function $G_0$ is always evaluated at the unperturbed ground state energy $E_0$, and will be non-positive. 

The product of operators obtained between between $\Upsilon$'s will vanish unless it remains within (or at least overlaps in some nontrivial way) the ground-space of every code gadget on the lattice. It will be simple to eliminate many operators which will vanish or act trivially on the ground space. At $n^{\mathrm{th}}$ order, the self-energy will consist of a sum of terms
\begin{equation}
\Sigma^{(n)}(E_0) = \sum_j\lambda^n\Upsilon \kappa_jA_j\Upsilon
\end{equation}
with $\kappa_j$ constants. They take into account the Green's function terms appearing in the perturbative analysis. Summed over each possible ordering of the $V$ terms which will give the same $A_j$, $\kappa_j$ is the product of the $(n-1)$ Green's functions appearing in the $n^{\mathrm{th}}$ order terms. In our case, this sum will run over the $n!$ ways of ordering the perturbations which multiply together to give the operator $A_j$.
%
%

\section{Extension to Arbitrary Graph}

It is not difficult to extend our treatment from the explicit square lattice to an arbitrary directed graph. As each edge is associated with two plaquettes on any nonintersecting 2D graph (neglecting boundaries), each edge qudit is associated with exactly 4 nearest neighbours regardless of the form of the lattice. This allows us to retain our previous definitions for edge qubit gauge and logical operators as in \Sref{sec-naencoding}.

If we apply a perturbation to the uncoupled edge qudits exactly as before, we will again see plaquette and vertex terms arising in the effective Hamiltonian. Of course, they may not arise at the same order in this treatment (e.g. for a hexagonal lattice plaquettes arise at $6^{\mathrm{th}}$ and vertices arise at $3^{\mathrm{rd}}$ order, and on a general graph plaquette boundaries and vertex stars will not be uniform in size). It is possible that this may have some undesirable effect on excited states of the effective Hamiltonian, but the higher order terms that will be able to survive the perturbation will act in the logical space as products of existing (commuting) terms until the perturbative order is high enough to form non-trivial homology cycles over the lattice. Below this order, the ground space of the effective Hamiltonian will remain within the ground space of the encoded quantum double Hamiltonian.

\begin{figure}
\centering
\beginpgfgraphicnamed{Figures/simpleDirectedGraph}
\begin{tikzpicture}[scale=2.5]
	\def\off{0.5} 
	\draw (0,0)--(1,0);
	\draw[->,>=stealth',line width=2.5](\off,0)-- (0.1+\off,0); 
	\draw (1,0)--(1.87,0.5);
	\draw[->,>=stealth',line width=2.5](1+\off*0.87,\off*0.5)-- (1.087+\off*0.87,0.05+\off*0.5); 
	\draw (1,0)--(1.87,-0.5);
	\draw[->,>=stealth',line width=2.5](1+\off*0.87,-\off*0.5)-- (1.087+\off*0.87,-0.05-\off*0.5); 
	\draw (1.87,0.5)--(2.87,0.5);
	\draw[->,>=stealth',line width=2.5](1.87+\off,0.5)-- (1.97+\off,0.5); 
	\draw (1.87,-0.5)--(2.87,-0.5);
	\draw[->,>=stealth',line width=2.5](1.87+\off,-0.5)-- (1.97+\off,-0.5); 
	\draw (2.87,0.5)--(2.87,-0.5);
	\draw[->,>=stealth',line width=2.5](2.87,0.5-\off)-- (2.87,0.4-\off); 
	\node at (1,0) [anchor = north east] {$v_1$};
	\node at (2.2,0) [anchor = center] {$p_1$};
\end{tikzpicture}
\endpgfgraphicnamed
\caption{A section of a simple directed graph\label{F:ArbGraph}}
\end{figure}
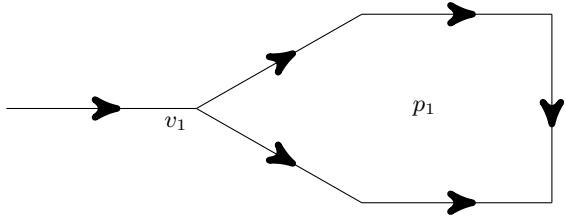

We can see in \Fref{F:ArbGraph} part of a simple directed graph. If we consider the perturbation term exactly as in the general group square lattice treatment (\Eref{eqnapert}), the vertex term will take the form (up to additive and multiplicative constants)
\begin{equation}
H(v) \sim \Upsilon\sum_g\prod_{e\rightarrow v}L^g_{\subfont{L}+}\prod_{e\leftarrow v}L^g_{\subfont{L}-} \Upsilon+ \Upsilon\sum\prod_{e\rightarrow v}Z_{\subfont{G}'+}\prod_{e\leftarrow v}Z_{\subfont{G}'-}\Upsilon
\end{equation}
where $e\leftarrow v$ denotes those edges which run out of the vertex $v$, and $e\rightarrow v$ denotes those that run towards the vertex. The exact form of the $Z$ term is not difficult to calculate but is only sketched in this notation for clarity. In any case, the vertex $Z$ term will evaluate as a constant in the ground space and so can be disregarded for our purposes. This leaves the products of terms for inwards directed edges and outwards directed edges. If we look back to the definitions of the quantum double model in \Sref{secprelimqd}, we can see that this will give a vertex term consistent with the target Hamiltonian.

Similarly, as we traverse a plaquette, two kinds of terms will arise,
\begin{equation}
H(v) \sim \Upsilon\sum_g\prod_{e\rightarrow p}L^g_{\subfont{G}-}\prod_{e\leftarrow p}L^g_{\subfont{G}+} \Upsilon+ \Upsilon\sum\prod_{e\rightarrow v}Z_{\subfont{L}+}\prod_{e\leftarrow v}Z_{\subfont{L}-}\Upsilon
\end{equation}

Here the notation $e\rightarrow p$ is taken to mean the edge is on the left of the plaquette, when looking along $e$, and similarly $e\leftarrow p$ has the edge on the right of $p$. The indices of the $Z$ terms are here suppressed for clarity, but they can easily be restored by comparing with \Eref{eq-finalnaeffham}, and can be verified to be consistent with the quantum double Hamiltonian. Of course, the gauge $L$ term will be evaluated as a constant and disregarded. This will leave us only with the desired logical operators on both plaquettes and vertices, and we will have successfully reproduced the quantum double Hamiltonian (with the caveat that the terms may arise at different order, and with different coefficients).

\section{Proof of Theorem \ref{thmabhamgs}}\label{secProveThmAbGs}
\begin{proof}
We will label eigenstates of $Z$ by group elements $\ket{h}$, and the corresponding eigenstates of $L$ by (unitary, irreducible) representations $\ket{\sigma} = \frac{1}{\sqrt{d}}\sum_h\omega^{\sigma h}\ket{h}$. Note that for these cyclic groups, the representations are all one dimensional. We begin by defining projectors
\begin{equation}
	P_i^k = \frac{1}{d}\sum_l\omega^{kl}S_i^l \, ,
\end{equation}
with inverses
\begin{equation}
	S_i^k = \sum_l\omega^{-kl}P^l_i \, .
\end{equation}
where $i \in\{L,Z\}$. We can then write our Hamiltonian as
\begin{equation}
	H=-\frac{1}{d}\sum_k \left[L_\subfont{G}^k 
	\left(1+\sum_h\omega^{-kh}P_L^h\right)
	+Z_\subfont{G}^k \left(1+\sum_{\sigma}\omega^{-k\sigma}P_Z^\sigma\right)\right] \, .
\end{equation}

If we block diagonalize with respect to both $h$ and $\sigma$, our proof will amount to showing that the unique ground state is in the $h=0$ and $\sigma=0$ block. In group theoretic terms, this corresponds to $h$ the identity group element and $\sigma$ the trivial representation (which we wil denote now as $\sigma=I$). Our Hamiltonian within the $(h,\sigma)$ block takes the form
\begin{equation}
	H^{h,\sigma}=-\frac{1}{d}\sum_k
	\left[L_\subfont{G}^k(1+\omega^{-kh})+Z_\subfont{G}^k(1+\omega^{-k\sigma})\right] \, .
\end{equation} 

We can rewrite the $\sum_kL_\subfont{G}^k$ and $\sum_kZ_\subfont{G}^k$ operators as projectors on the gauge subspace onto group elements or representations:
\begin{equation}
	\sum_k\omega^{-kh}L_\subfont{G}^k = d\proj{h} \quad \mbox{ and } \quad 
	\sum_k\omega^{-k\sigma}Z_\subfont{G}^k = d\proj{\sigma} \, ,
\end{equation}
with $\ket{\sigma}$ a representation state and $\ket{h}$ a group element state. The representation states are discrete Fourier transforms of the element states. Then we have
\begin{align}\label{E:blockHableian}
	H^{h,\sigma} = -\bigl[\proj{0}+\proj{h}+\proj{I}+\proj{\sigma}\bigr] \, .
\end{align}

Here $\ket{0}$ is the group identity element and $\ket{I}$ corresponds to the trivial representation of the group.  Since $H^{h,\sigma}$ is a negative sum of projectors, all its eigenvalues are non-positive.  We can then write a triangle inequality for the magnitude of least eigenvalue by taking the operator norm,
\begin{align}
	\bigl\|H^{h,\sigma}\bigr\| \leq \bigl\|\proj{0}+\proj{I}\bigr\|
		+\bigl\|\proj{h}+\proj{\sigma}\bigr\| \, .
\end{align}

We know that this inequality is saturated only if $\left(\proj{0}+\proj{I}\right)$ is parallel to $\left(\proj{h}+\proj{\sigma}\right)$ in their largest eigenspace.  But unless $(h,\sigma)=(0,I)$, these vectors are never parallel.  Then we have
\begin{align}
	\bigl\|H^{0,I}\bigr\| = 2\bigl\| \proj{0}+\proj{I} \bigr\| > \bigl\|H^{h,\sigma}\bigr\|
	\quad \mbox{ when } (h,\sigma) \neq (0,I) \, .
\end{align}

We can now say that the ground space must have $h=0$ and $\sigma=I$. Furthermore, when this is the case, we see from \Eref{E:blockHableian} that the only non-vanishing eigenvectors of this block are linear combinations of $\ket{0}$ and $\ket{I}$.  We find by inspection that the two independent eigenvectors are
\begin{align}
	\ket{\psi_\pm} = \ket{0}\pm\ket{I} \, ,
\end{align}
with eigenvalues
\begin{equation}
	\lambda_{\pm} = -2\left[1\pm\frac{1}{\sqrt{d}}\right] \, .
\end{equation}

This gives for the unique ground state of the gauge qudit
\begin{equation}
	\ket{\psi_0^\subfont{G}} = \ket{0}+\ket{I} \, .
\end{equation}

Given that the stabilizer and gauge degrees of freedom have a unique ground state, and the fact that the Hamiltonian commutes with the logical operators, we can also say that the Hamiltonian \Eref{eqabham1} has a $d$-fold degenerate ground space encoding the logical state.\qedhere
\end{proof}


%
%

\section{Proof of Theorem \ref{thmnags}}\label{provethmnags}

\begin{proof}
We wish to prove that our Hamiltonian \Eref{eqnahamunpet} is stabilized as claimed. The key to proving this theorem is choosing a suitable basis. We have previously defined the states as in the physical basis of \Eref{eq-physicalbasis}. Here we will need to look at some more sophisticated bases that make the actions of our encoded operators even more transparent.
\subsection{Alternative Bases}

Consider first the gauge representation basis:
\begin{equation}
\ket{\sigma_{mn},h,\pi_{ij},k_\subfont{L}} \equiv \frac{\sqrt{d_\pi d_\sigma}}{|G|}\sum_{g_\pi, g_\sigma \in G}[\pi(g_\pi g_\sigma)]_{ij}[\sigma(g_\sigma)]_{mn}\ket{g_\sigma,g_\sigma^{-1}hg_\sigma,g_\pi,k_\subfont{L}}
\end{equation}

Here $\pi$ and $\sigma$ are (unitary irreducible) representations of dimension $d_\pi$ and $d_\sigma$ respectively. The $\pi_{ij}$ variable defines the gauge state, which is why we call this the gauge representation basis. It is clear that there are a total of $|G|^4$ basis states, as needed, and orthonormality can be easily verified. We would like to know how the $\sum_g L^g_{\subfont{G}\pm}$ projectors in the Hamiltonian \Eref{eqnahamunpet} will act on this basis.

It can be easily verified that these operators act in the following way on the physical basis
\begin{eqnarray}
L^{\tilde{g}}_{\subfont{G}+}\left|h_a,h_b,k_{\subfont{G}},k_{\subfont{L}}\right>&=&\left|h_a\tilde{g}^{-1},\tilde{g}h_b\tilde{g}^{-1},k_{\subfont{G}}h_a\tilde{g}h_a^{-1},k_{\subfont{L}}\right>\\
L^{\tilde{g}}_{\subfont{G}-}\left|h_a,h_b,k_{\subfont{G}},k_{\subfont{L}}\right> &=& \left|h_a,h_b,k_{\subfont{G}}h_a\tilde{g}^{-1}h_a^{-1},k_{\subfont{L}}\right> .
\end{eqnarray}

This gives for the action of the $L$ projectors in our new basis:
\begin{equation}
\frac{1}{|G|}\sum_{\tilde{g}}L^{\tilde{g}}_{\subfont{G}+}\ket{\sigma_{mn},h,\pi_{ij},k_\subfont{L}}=\frac{\sqrt{d_\pi d_\sigma}}{|G|^2}\sum_{g_\pi, g_\sigma,\tilde{g}}[\pi(g_\pi g_\sigma)]_{ij}[\sigma(g_\sigma)]_{mn}\ket{g_\sigma\tilde{g}^{-1},\tilde{g}g_\sigma^{-1}hg_\sigma\tilde{g}^{-1},g_\pi g_\sigma\tilde{g}g_\sigma^{-1},k_\subfont{L}}
\end{equation}
introducing substitutions
\begin{eqnarray}
g_\sigma' &=& g_\sigma\tilde{g}^{-1}\\
g_\pi'&=&g_\pi g_\sigma\tilde{g}g_\sigma^{-1}
\end{eqnarray}
we find
\begin{equation}
\frac{1}{|G|}\sum_{\tilde{g}}L^{\tilde{g}}_{\subfont{G}+}\ket{\sigma_{mn},h,\pi_{ij},k_\subfont{L}}=\frac{\sqrt{d_\pi d_\sigma}}{|G|^2}\sum_{g_\pi', g_\sigma',\tilde{g}}\sum_s[\pi(g_\pi'g_\sigma')]_{ij}[\sigma(g_\sigma')]_{ms}[\sigma(\tilde{g})]_{sn}\ket{g_\sigma',g_\sigma'^{-1}hg_\sigma',g_\pi',k_\subfont{L}}
\end{equation}

The grand orthogonality theorem (\Eref{eq-orthogonality}) can then be used over $\tilde{g}$ to give
\begin{eqnarray}
\frac{1}{|G|}\sum_{\tilde{g}}L^{\tilde{g}}_{\subfont{G}+}\ket{\sigma_{mn},h,\pi_{ij},k_\subfont{L}}
&=& \delta_{\sigma I}\frac{\sqrt{d_\pi d_\sigma}}{|G|}\sum_{g_\pi', g_\sigma'}[\pi(g_\pi'g_\sigma')]_{ij}[\sigma(g_\sigma')]_{11}\ket{g_\sigma',g_\sigma'^{-1}hg_\sigma',g_\pi',k_\subfont{L}}\\
&=&\delta_{\sigma I}\ket{\sigma_{mn},h,\pi_{ij},k_\subfont{L}}
\end{eqnarray}
where $I$ is the trivial irrep of the group $G$. We can perform a similar treatment for the other $L$ projector, giving
\begin{equation}
\frac{1}{|G|}\sum_{\tilde{g}}L^{\tilde{g}}_{\subfont{G}-}\ket{\sigma_{mn},h,\pi_{ij},k_\subfont{L}}
=\frac{\sqrt{d_\pi d_\sigma}}{|G|^2}\sum_{g_\pi ', g_\sigma ,\tilde{g}}\sum_s[\pi(g_\pi 'g_\sigma )]_{is}[\pi(\tilde{g}^{-1})]_{sj}[\sigma(g_\sigma )]_{mn}\ket{g_\sigma ,g_\sigma ^{-1}hg_\sigma ,g_\pi ',k_\subfont{L}}
\end{equation}
with substitutions
\begin{eqnarray}
g_\pi '&=&g_\pi g_\sigma \tilde{g}g_\sigma ^{-1}
\end{eqnarray}

Again, this can be simplified through orthogonality to give
\begin{eqnarray}
\frac{1}{|G|}\sum_{\tilde{g}}L^{\tilde{g}}_{\subfont{G}-}\ket{\sigma_{mn},h,\pi_{ij},k_\subfont{L}}
&=&\delta_{\pi I}\frac{\sqrt{d_\pi d_\sigma }}{|G|}\sum_{g_\pi ', g_\sigma }[\pi(g_\pi 'g_\sigma )]_{11}[\sigma(g_\sigma )]_{mn}\ket{g_\sigma ,g_\sigma ^{-1}hg_\sigma ,g_\pi ',k_\subfont{L}}\\
&=&\delta_{\pi I}\ket{\sigma_{mn},h,\pi_{ij},k_\subfont{L}}
\end{eqnarray}

To treat the $T_{G\pm}$ operators analogously, we need to look at another basis in which the gauge state takes a particular goup element as opposed to an element of a representation. We define this gauge element basis as
\begin{equation}
\ket{\sigma_{mn},h,k_\subfont{G},k_\subfont{L}} \equiv \frac{\sqrt{d_\sigma }}{\sqrt{|G|}}\sum_{g_\sigma  \in \mathcal{G}}[\sigma(g_\sigma )]_{mn}\ket{g_\sigma ,g_\sigma ^{-1}hg_\sigma ,k_\subfont{G},k_\subfont{L}}
\end{equation}
i.e. instead of using a representation for the gauge degree of freedom, here a single group element is used. In terms of the physical basis, the $T^{1}_{G\pm}$ projectors act in the following way:
\begin{eqnarray}
T^{1}_{\subfont{G}+}\left|h_a,h_b,k_{\subfont{G}},k_{\subfont{L}}\right>&=&\delta_{1k_{\subfont{G}}}\left|h_a,h_b,k_{\subfont{G}},k_{\subfont{L}}\right> \label{eq:gadgetmagparts1}\\
T^{1}_{\subfont{G}-}\left|h_a,h_b,k_{\subfont{G}},k_{\subfont{L}}\right> &=& \delta_{(h_ah_bh_a^{-1})k_{\subfont{G}}}\left|h_a,h_b,k_{\subfont{G}},k_{\subfont{L}}\right> \label{eq:gadgetmagparts2}
\end{eqnarray}
where $1$ is the identity group element. We are now in a position to demonstrate that the $T$ operators explicitly project to single states in the gauge element basis:
\begin{eqnarray}
T^{1}_{\subfont{G}+}\ket{\sigma_{mn},h,k_\subfont{G},k_\subfont{L}}&=&\delta_{1k_{\subfont{G}}}\frac{\sqrt{d_\sigma }}{\sqrt{|G|}}\sum_{g_\sigma }[\sigma(g_\sigma )]_{mn}\ket{k_\subfont{G},\sigma_{mn},h,k_\subfont{L}}\\
&=&\delta_{ek_{\subfont{G}}}\ket{R_{mn},h,k_\subfont{G},k_\subfont{L}}\\
T^{1}_{\subfont{G}-}\ket{\sigma_{mn},h,k_\subfont{G},k_\subfont{L}}&=&\delta_{(g_\sigma g_\sigma ^{-1}hg_\sigma g_\sigma ^{-1})k_{\subfont{G}}}\frac{\sqrt{d_\sigma }}{\sqrt{|G|}}\sum_{g_\sigma }[\sigma(g_\sigma )]_{mn}\ket{g_\sigma ,g_\sigma ^{-1}hg_\sigma ,k_\subfont{G},k_\subfont{L}}\\
&=&\delta_{hk_{\subfont{G}}}\ket{\sigma_{mn},h,k_\subfont{G},k_\subfont{L}}
\end{eqnarray}

\subsection{Solving the Hamiltonian}
The Hamiltonian of interest is given by
\begin{equation}
H=-\sum_{g}L_{\subfont{G}+}^g-\sum_gL_{\subfont{G}-}^g-T^1_{\subfont{G}+}-T^1_{\subfont{G}-}
\end{equation}

Given that we now know explicitly what states each of these terms project to, we can write instead
\begin{eqnarray}
H &=& -\sum_{\pi_{ij},h,k_\subfont{L}}\proj{I_{11},h,\pi_{ij},k_\subfont{L}}-\sum_{\sigma_{mn},h,k_\subfont{L}}\proj{\sigma_{mn},h,I_{11},k_\subfont{L}}\\
&&\qquad\qquad-\sum_{\sigma_{mn},h,k_\subfont{L}}\proj{\sigma_{mn},h,1,k_\subfont{L}} -\sum_{\sigma_{mn},h,k_\subfont{L}}\proj{\sigma_{mn},h,h,k_\subfont{L}}\nonumber
\end{eqnarray}

Each of these terms are pairwise orthogonal for different values of $h$ or $k_\subfont{L}$. We can then immediately block diagonalize the Hamiltonian by these two variables. As these labels will not participate overly in the calculation, we will suppress them henceforth for clarity. We can also block diagonalize by a representation and one of its indices as follows
\begin{equation}
H^{h,k_\subfont{L}}_{\sigma,n} = -\sum_{m}\Big[\proj{I_{11},\sigma_{mn}}+\proj{\sigma_{mn},I_{11}}+\proj{\sigma_{mn},1}+\proj{\sigma_{mn},h}\Big]
\end{equation}

That projectors in each block are orthogonal to those in different blocks can be verified by examining the inner products presented in \Sref{sec-nagsproofIP}.

Looking only in one block of the Hamiltonian, we can split the terms into two vectors:
\begin{equation}
-H^{h,k_\subfont{L}}_{\sigma,n} = \Big[\sum_m\proj{I_{11},\sigma_{mn}}+\sum_m\proj{\sigma_{mn},1}\Big]+\Big[\sum_m\proj{\sigma_{mn},I_{11}}+\sum_m\proj{\sigma_{mn},h}\Big]
\end{equation}

We can bound the length of each of these vectors individually by considering them as subhamiltonians
\begin{eqnarray}
-H_A &=& \sum_m\proj{I_{11},\sigma_{mn}}+\sum_m\proj{\sigma_{mn},1}\\
-H_B&=&\sum_m\proj{\sigma_{mn},I_{11}}+\sum_m\proj{\sigma_{mn},h}
\end{eqnarray}
each of these subhamiltonians can be further block diagonalized by $m$:
\begin{eqnarray}
-H_A^m &=& \proj{I_{11},\sigma_{mn}}+\proj{\sigma_{mn},1}\\
-H_B^m &=& \proj{\sigma_{mn},I_{11}}+\proj{\sigma_{mn},h}
\end{eqnarray}

We can solve these subhamiltonians in each block by calculating the overlap of the two projectors (these calculations are shown in \Sref{sec-nagsproofIP}:
\begin{eqnarray}
\braket{I_{11},\sigma_{mn}}{\sigma_{mn},1} &=& \frac{1}{\sqrt{|G|}}\\
\braket{\sigma_{mn},I_{11}}{\sigma_{mn},h} &=& \frac{1}{\sqrt{|G|}}
\end{eqnarray}

It can be shown that the eigenvalues of these subhamiltonians will then be
\begin{equation}
\lambda_m = -1\pm|C|
\end{equation}
where $C$ is the inner product calculated above. As such, we can say that the norm of the subhamiltonians are equal, independent of $m$ block and have value $1+\frac{1}{\sqrt{|G|}}$. Importantly, their length is independent of the choice of $\sigma_{mn}$, $\pi_{ij}$, $h$ and $k_\subfont{L}$.

In terms of the full Hamiltonian and block structure, we have two vectors whose lengths are constant for a given group. We can then bound the eigenvalues of the full Hamiltonian as follows. We can use a triangle inequality argument to give
\begin{equation}
\Big\|H^{h,k_\subfont{L}}_{\sigma,n}\Big\| \leq \Big\|\sum_m\proj{I_{11},\sigma_{mn}}+\sum_m\proj{\sigma_{mn},1}\Big\|+\Big\|\sum_m\proj{\sigma_{mn},I_{11}}+\sum_m\proj{\sigma_{mn},h}\Big\|
\end{equation}

We know that each individual vector has a constant length. We also know that this inequality will only be saturated when the two vectors are parallel (which will only happen when they are equal). This is also when $\sigma=I$, $n=1$, and $h=e$. As desired, there is no dependence on $k_\subfont{L}$. 

Since the eigenvectors of $H$ are non-positive, the largest magnitude eigenvalue will correspond to the ground space. The triangle inequality argument amounts to showing that the ground state must rest in the block defined by $H^{1,k_\subfont{L}}_{I,1}$. By inspection, there will be 2 independent eigenvectors in this block for each value of $k_\subfont{L}$ (we return the $h$ index, but the $k_\subfont{L}$ index remains suppressed here)
\begin{equation}
\ket{\psi_\pm} = \ket{I_{11},1,1} \pm \ket{I_{11},1,I_{11}}
\end{equation}
with eigenvalues
\[
\lambda_{\pm} =-2\left(1\pm\frac{1}{\sqrt{|G|}}\right)
\]

Clearly $\ket{\psi_+}$ defines a $|G|$-fold degenerate ground space of our system. It is simple to verify that the stabilizers have the desired relations with this state
\begin{eqnarray}
\sum_gS_L^g\ket{\psi_+}&=&\ket{\psi_+}\\
S_T^1\ket{\psi_+}&=&\ket{\psi_+}
\end{eqnarray}
and thus the Hamiltonian \Eref{eqnahamunpet} behaves as claimed.
\qedhere\end{proof}

\subsection{Useful Inner Products}\label{sec-nagsproofIP}
Here we wish to calculate inner products between states in the gauge representation basis and the gauge element basis. In general, we have
\begin{eqnarray}
\braket{\sigma_{mn},h,\pi_{ij},k_\subfont{L}}{\sigma'_{m'n'},h',k_\subfont{G},k_\subfont{L}'}&=&\frac{\sqrt{d_\pi  d_{\sigma'}d_\sigma }}{|G|^{3/2}}\sum_{g_\pi , g_\sigma , g_{\sigma'}}[\pi(g_\pi g_\sigma )]^*_{ij}[\sigma(g_\sigma )]^*_{mn}[\sigma'(g_{\sigma'})]_{m'n'}\nonumber\\
&&\qquad\qquad \times\braket{g_\sigma ,g_\sigma ^{-1}hg_\sigma ,g_\pi ,k_\subfont{L}}{g_{\sigma'},g_{\sigma'}^{-1}h'g_{\sigma'},k_\subfont{G},k_\subfont{L}'}\\
&=&\frac{\sqrt{d_\pi  d_{\sigma'}d_\sigma }}{|G|^{3/2}}\sum_{g_\pi , g_\sigma , g_{\sigma'}}[\pi(g_\pi g_\sigma )]^*_{ij}[\sigma(g_\sigma )]^*_{mn}\sigma'(g_{\sigma'})]_{m'n'}\nonumber\\
&&\qquad\qquad\times\delta_{g_\sigma g_{\sigma'}}\delta_{hh'}\delta_{g_\pi k_\subfont{G}}\delta_{k_\subfont{L}k_\subfont{L}'}\\
&=&\frac{\sqrt{d_\pi  d_{\sigma'}d_\sigma }}{|G|^{3/2}}\delta_{hh'}\delta_{k_\subfont{L}k_\subfont{L}'}\sum_{g_\sigma }[\pi(k_\subfont{G}g_\sigma )]^*_{ij}[\sigma(g_\sigma )]^*_{mn}[\sigma'(g_{\sigma})]_{m'n'}\label{eq-generalresult}
\end{eqnarray}

In particular, the special cases we are interested can be simplified. With $I$ the trivial representation, we have
\begin{eqnarray}
\braket{\sigma_{mn},h,I_{11},k_\subfont{L}}{\sigma'_{m'n'},h',k_\subfont{G},k_\subfont{L}'}
&=&\frac{\sqrt{d_{\sigma'}d_\sigma }}{|G|^{3/2}}\delta_{hh'}\delta_{k_\subfont{L}k_\subfont{L}'}\sum_{g_\sigma }[\sigma(g_\sigma )]^*_{mn}[\sigma'(g_{\sigma})]_{m'n'}\\
&=&\frac{1}{|G|^{1/2}}\delta_{hh'}\delta_{k_\subfont{L}k_\subfont{L}'}\delta_{\sigma \sigma'}\delta_{mm'}\delta_{nn'}\label{eq-overlap1}\qquad\forall k_\subfont{G}
\end{eqnarray}
and
\begin{eqnarray}
\braket{I_{11},h,\pi_{ij},k_\subfont{L}}{\sigma'_{m'n'},h',k_\subfont{G},k_\subfont{L}'}
&=&\frac{\sqrt{d_\pi  d_{\sigma'}}}{|G|^{3/2}}\delta_{hh'}\delta_{k_\subfont{L}k_\subfont{L}'}\sum_{g_\sigma }[\pi(k_\subfont{G}g_\sigma )]^*_{ij}[\sigma'(n_{R})]_{m'n'}\\
&=&\frac{\sqrt{d_\pi  d_{\sigma'}}}{|G|^{3/2}}\delta_{hh'}\delta_{k_\subfont{L}k_\subfont{L}'}\sum_{g_\sigma }\sum_s[\pi(k_\subfont{G})]^*_{is}[\pi(g_\sigma )]^*_{sj}[\sigma'(n_{R})]_{m'n'}\\
&=&\frac{1}{|G|^{1/2}}\delta_{hh'}\delta_{k_\subfont{L}k_\subfont{L}'}\delta_{\pi\sigma'}\delta_{jn'}[\pi(k_\subfont{G})]^*_{im'}\label{eq-badresult}
\end{eqnarray}
particularly,
\begin{eqnarray}
\braket{I_{11},h,\pi_{ij},k_\subfont{L}}{\sigma'_{m'n'},h',1,k_\subfont{L}'}
&=&\frac{1}{|G|^{1/2}}\delta_{hh'}\delta_{k_\subfont{L}k_\subfont{L}'}\delta_{GR'}\delta_{im'}\delta_{jn'}\label{eq-overlap2}\qquad\forall k_\subfont{G}
\end{eqnarray}

The orthonormality of the gauge representation basis is also useful when written in the form
\begin{equation}
\braket{I_{11},h,\pi_{ij},k_\subfont{L}}{\sigma_{mn},h',I_{11},k_\subfont{L}'}=\delta_{\pi \sigma}\delta_{\sigma I}\delta_{i1}\delta_{nj}\delta_{m1}\delta_{n1}\delta_{hh'}\delta_{k_\subfont{L}k_\subfont{L}'}
\end{equation}

%
%

\section{Error Operations in General Quantum Double Models}\label{secerrorops}

In the case of cyclic quantum double models, we were able to provide a simple general argument as to why many of the terms in the effective Hamiltonian vanished. In the more general case, this kind of simple argument is no longer applicable, and so we are forced to take an exhaustive survey of the terms that may arise in the effective Hamiltonian. We aim to show that any undesired terms will vanish in the ground space.

In order to undertake this study, it is useful to note the following relations:
\begin{eqnarray}
Z_{\pm}^{\pi_{ij}}L_{\pm}^g &=& \sum_k [\pi(g)]_{ik}L_{\pm}^gZ_{\pm}^{\pi_{kj}}\\
Z_{\pm}^{\pi_{ij}}L_{\mp}^g &=& \sum_k L_{\mp}^gZ_{\pm}^{\pi_{ik}}[\pi(g^{-1})]_{kj}
\end{eqnarray}

The kinds of operators we consider will take the form:
\begin{align}
z_a =& \siteNXL{Z^{\pi_{ij}}_-}{Z^{\sigma_{kl}}_+}{I}{I} \ , & 
z_b =& \siteNXL{I}{I}{Z^{\sigma_{kl}}_+}{Z^{\pi_{ij}}_-} \ , &
z_c =& \siteNXL{I}{Z^{\pi_{ij}}_+}{Z^{\sigma_{kl}}_+}{I} \ , &
z_d =& \siteNXL{Z^{\pi_{ij}}_-}{I}{I}{Z^{\sigma_{kl}}_-}\nonumber \\
l_a =& \siteN{L^{g}_+}{L^{g'}_+}{I}{I} \ , &
l_b =& \siteN{I}{I}{L^{g'}_-}{L^{g}_-} \ , &
l_c =& \siteN{I}{L^{g}_+}{L^{g'}_-}{I} \ , &
l_d =& \siteN{L^{g}_+}{I}{I}{L^{g'}_-}
\end{align}

We will refer to these as ``error terms'' or ``error operators''. For each error term $\hat{E}$ we will calculate $\Upsilon \hat{E}\Upsilon$ where $\Upsilon$ is the projector on to the ground state. Note that $\Upsilon = S\Upsilon$ for either of the stabilizers $S_T$ or $S_{L}$. First consider the error term $z_c$:
\begin{eqnarray}
\Upsilon z_c\Upsilon &=& \Upsilon z_cS_L\Upsilon\\
&=&\frac{1}{|G|}\sum_g\Upsilon z_cS^g_L\Upsilon\\
&=&\frac{1}{|G|}\sum_g\Upsilon \siteGIGANTOR{L^{g^{-1}}_{-}}{Z^{\pi_{ij}}_+L^{g^{-1}}_{-}}{Z^{\sigma_{kl}}_+L^{g^{-1}}_{+}}{L^{g^{-1}}_{+}}\Upsilon\\
&=&\frac{1}{|G|}\sum_g\Upsilon (S^{g^{-1}}_L)^{\dagger}\sum_{m,n}[\pi(g)]_{mj}[\sigma(g^{-1})]_{kn}\siteNXL{I}{Z^{\pi_{im}}_+}{Z^{\sigma_{nl}}_+}{I}\Upsilon
\end{eqnarray}

Now, because the ground state is stabilized by the projector $S_{L}$, it is easy to check that it must also be stabilized by each of $S_L^g$ individually. This allows us to write:
\begin{eqnarray}
\Upsilon z_c\Upsilon &=& \frac{1}{|G|}\Upsilon \sum_{m,n,g}[\pi(g)]_{mj}[\sigma(g^{-1})]_{kn}\siteNXL{I}{Z^{\pi_{im}}_+}{Z^{\sigma_{nl}}_+}{I}\Upsilon\\
&=& \frac{1}{|G|}\Upsilon \sum_{m,n,g}[\pi(g)]_{mj}[\sigma(g)]^*_{nk}\siteNXL{I}{Z^{\pi_{im}}_+}{Z^{\sigma_{nl}}_+}{I}\Upsilon
\end{eqnarray}
where we have used the unitarity of the representation $\sigma$ in the last line. Now we make use of the Grand Orthogonality Theorem for unitary representations:
\begin{equation}\label{eq-orthogonality}
\sum_g[\pi(g)]^*_{ij}[\sigma(g)]_{kl} = \frac{|G|}{d_{\pi}}\delta_{\pi\sigma}\delta_{ik}\delta_{jl}
\end{equation}
and we can write
\begin{eqnarray}
\Upsilon z_c\Upsilon &=&\Upsilon \sum_{m,n}\frac{1}{d_{\pi}}\delta_{\pi\sigma}\delta_{mn}\delta_{jk}\siteNXL{I}{Z^{\pi_{im}}_+}{Z^{\sigma_{nl}}_+}{I}\Upsilon\\
&=& \Upsilon \sum_{m}\frac{1}{d_{\pi}}\delta_{\pi\sigma}\delta_{jk}\siteNXL{I}{Z^{\pi_{im}}_+}{Z^{\sigma_{ml}}_+}{I}\Upsilon\label{zerroreq1}
\end{eqnarray}

This gives us the result that all errors of this type will vanish from the effective Hamiltonian unless they satisfy the conditions $\pi = \sigma$ and $j=k$. It also allows us to rewrite an operator of the form $z_c$ satisfying these conditions as an average over the shared index in the ground state. This will become important when we undertake the perturbation calculations.

We can find analogous results for the other $z$ error terms,
\begin{eqnarray}
\Upsilon z_d\Upsilon &=& \Upsilon \sum_{m,n,g}[\pi(g^{-1})]_{im}[\sigma(g)]_{nl}\siteNXL{Z^{\pi_{mj}}_-}{I}{I}{Z^{\sigma_{kn}}_-}\Upsilon\\
&=& \Upsilon \sum_{m}\frac{1}{d_{\pi}}\delta_{\pi\sigma}\delta_{il}\siteNXL{Z^{\pi_{mj}}_-}{I}{I}{Z^{\sigma_{km}}_-}\Upsilon\label{zerroreq2}
\end{eqnarray}
\begin{eqnarray}
\Upsilon z_a\Upsilon &=& \Upsilon \sum_{m,n,g}[\pi(g^{-1})]_{im}[\sigma(g)]_{nl}\siteNXL{Z^{\pi_{mj}}_-}{Z^{\sigma_{kn}}_+}{I}{I}\Upsilon\\
&=& \Upsilon \sum_{m}\frac{1}{d_{\pi}}\delta_{\pi\sigma}\delta_{il}\siteNXL{Z^{\pi_{mj}}_-}{Z^{\sigma_{km}}_+}{I}{I}\Upsilon
\end{eqnarray}
\begin{eqnarray}
\Upsilon z_b\Upsilon &=& \Upsilon \sum_{m,n,g}[\pi(g)]_{mj}[\sigma(g^{-1})]_{kn}\siteNXL{I}{I}{Z^{\sigma_{nl}}_+}{Z^{\pi_{im}}_-}\Upsilon\\
&=& \Upsilon \sum_{m}\frac{1}{d_{\pi}}\delta_{\pi\sigma}\delta_{jk}\siteNXL{I}{I}{Z^{\sigma_{ml}}_+}{{Z^{\pi_{im}}_-}}\Upsilon\label{zerroreq4}
\end{eqnarray}

The $l$ error terms require slightly different treatment. Consider first $l_a$:
\begin{eqnarray}
\Upsilon l_a\Upsilon &=& \Upsilon l_aS^1_{T}\Upsilon\\
&=&\Upsilon \sum_{g_1g_2g_3g_4 = 1}\siteGIGANTOR{L^{g}_+T_{-}^{g_1}}{L^{g'}_+T_{+}^{g_2}}{T_{+}^{g_3}}{T_{-}^{g_4}}\Upsilon\\
&=&\Upsilon \sum_{g_1g_2g_3g_4 = 1} \siteNXL{T_{-}^{g_1g^{-1}}}{T_{+}^{g'g_2}}{T_{+}^{g_3}}{T_{-}^{g_4}}l_a\Upsilon\\
&=&\Upsilon \sum_{g_1gg'^{-1}g_2g_3g_4 = 1} \siteNXL{T_{-}^{g_1}}{T_{+}^{g_2}}{T_{+}^{g_3}}{T_{-}^{g_4}}l_a\Upsilon
\end{eqnarray}

The first operator is now orthogonal to the projector $S_T^{1}$ unless $g'g^{-1}=1$. Because the ground space is stabilized by $S_T^{1}$, we can say that this vanishes unless $g=g'$. This gives:
\begin{equation}
\Upsilon l_a\Upsilon = \Upsilon\delta_{gg'}\siteN{L^{g}_+}{L^{g'}_+}{I}{I}\Upsilon
\end{equation}
i.e. error terms of the form $l_a$ will vanish unless they obey $g=g'$. The procedure for $l_b$ proceeds similarly to yield:
\begin{eqnarray}
\Upsilon l_b\Upsilon &=& \Upsilon \sum_{g_1g_2g_3g'g^{-1}g_4 = 1} \siteNXL{T_{-}^{g_1}}{T_{+}^{g_2}}{T_{+}^{g_3}}{T_{-}^{g_4}}l_b\Upsilon\\
&=&\Upsilon \delta_{gg'}\siteN{I}{I}{L^{g'}_-}{L^{g}_-}\Upsilon
\end{eqnarray}

The error operators $l_c$ and $l_d$ are more complicated.
\begin{eqnarray}
\Upsilon l_c\Upsilon &=& \Upsilon l_c\sum_{g_1g_2g_3g_4 = 1} \siteNXL{T_{-}^{g_1}}{T_{+}^{gg_2}}{T_{+}^{g_3g'^{-1}}}{T_{-}^{g_4}}\Upsilon
\end{eqnarray}

We can imagine the second operator acting on state $\ket{h_a,h_b,k_{\subfont{G}},k_{\subfont{L}}}$. This leads to this operator vanishing unless
\[
h_a^{-1}g^{-1}k_\subfont{L}g'k_\subfont{L}^{-1}h_ah_b=1
\]

In the ground subspace, we know $h_b=1$, this implies
\begin{equation}
g = k_\subfont{L}g'k_\subfont{L}^{-1}
\end{equation}
i.e. the nonvanishing operators of the form $l_c$ depend explicitly on the logical state. A similar result is obtained for the error term $l_d$:
\begin{eqnarray}
\Upsilon l_d\Upsilon &=& \Upsilon l_d\sum_{g_1g_2g_3g_4 = 1} \siteNXL{T_{-}^{g_1g^{-1}}}{T_{+}^{g_2}}{T_{+}^{g_3}}{T_{-}^{g'g_4}}\Upsilon\\
\end{eqnarray}

The second operator vanishes unless 
\[
h_a^{-1}gk_\subfont{L}g'^{-1}k_\subfont{L}^{-1}h_ah_b=1
\]

In the codespace this is equivalent to:
\begin{eqnarray}
g = k_\subfont{L}g'k_\subfont{L}^{-1}
\end{eqnarray}

There are two remaining kinds of terms which will arise in our perturbation treatment. They may take a diagonal form, e.g. 
\[
\hat{E}\sim \siteN{\hat{A}}{I}{\hat{B}}{I}
\]
in which case they will not contribute to low order terms (and will have little impact on the form of higher order terms). A similar treatment to \Aref{secerrorops} and Eq.(\ref{eqopstudyfirst}-\ref{eqopstudylast}) reveals that these operators act as both gauge and logical operators, just as in the cyclic case.

Alternatively, error operators may consist of mixed $L$ and $Z$ operators, e.g.
\[
\hat{E}\sim  \siteNXL{Z^{\pi_{ij}}_-}{L^g_+}{I}{I}
\]

It is a simple calculation to show that these kinds of operators will vanish except when $\pi=I$ (the trivial representation) and $g=1$. That is, the only non-vanishing operator of this form is the identity operator. These results show that the encoding we use will indeed prevent undesirable excitations from being permitted, leaving only the terms from the target Hamiltonian to arise in our perturbative treatment.

\bibliographystyle{bibstyle}
\bibliography{main}

\end{document}